%% file: clifford.tex
\documentclass{LMCS}

\def\dOi{11(2:10)2015}
\lmcsheading%
{\dOi}
{1--17}
{}
{}
{Nov.~\phantom09, 2013}
{Jun.~19, 2015}
{}

\usepackage{tikz}
\usetikzlibrary{calc}
\usepackage{hyperref}

\input{clifford.sty}
\input{figures.tex}

\begin{document}

\title{Generators and relations for $n$-qubit Clifford operators}

\author{Peter Selinger}
\address{Department of Mathematics and Statistics, Dalhousie
  University}
\email{selinger@mathstat.dal.ca}
\thanks{Research supported by NSERC}


\ACMCCS{[{\bf Theory of computation}]: Models of computation---Quantum
  computation theory}

\begin{abstract}
  We define a normal form for Clifford circuits, and we prove that
  every Clifford operator has a unique normal form. Moreover, we
  present a rewrite system by which any Clifford circuit can be
  reduced to normal form. This yields a presentation of Clifford
  operators in terms of generators and relations.
\end{abstract}

\maketitle

\section{Introduction}

In quantum computation, an important subclass of quantum circuits is
the class of {\em Clifford circuits} or {\em stabilizer circuits}. It
is the smallest class of quantum circuits that includes the gates

\begin{equation}\label{eqn-generators}
\omega = e^{i\pi/4},\quad
H = \frac{1}{\sqrt{2}}\zmatrix{cc}{1&1\\1&-1},\quad
S = \zmatrix{cc}{1&0\\0&i},\quad
\displaystyle Z_c =
\zmatrix{cccc}{1&0&0&0\\0&1&0&0\\0&0&1&0\\0&0&0&-1},
\end{equation}
identities, and closed under the operations of tensor product and
composition. It is well-known that Clifford circuits can be
efficiently simulated on a classical computer, and therefore they are
not universal for quantum computing {\cite{Gottesman-1998}}. On the
other hand, Clifford gates are transversal in many quantum
error-correcting codes, and therefore they are particularly easy to
implement fault-tolerantly. For this reason, universal gate bases for
fault-tolerant quantum computation are often chosen to consist of the
Clifford operators and one additional gate, for example the
$\pi/8$-gate {\cite{Buhrman-Cleve-etal}}.

For all $n\geq 0$, the set of Clifford operators on $n$ qubits forms a
group, known as the {\em Clifford group} on $n$ qubits, which we
denote $\Clifford(n)$. It is well-known (and we will prove below) that the
Clifford group on $n$ qubits is finite and has
\[  |\Clifford(n)| = 8\cdot\prod_{i=1}^{n} 2(4^i-1)4^i
\]
elements; for example, the sizes for $n=1$, $2$, and $3$ are,
respectively:
\[  |\Clifford(1)| = 192,\quad
|\Clifford(2)| = 92160,\quad
|\Clifford(3)| = 743178240.
\]
In this paper, we define a normal form for Clifford circuits, and we
prove that every Clifford operator has a unique normal form. Moreover,
we present a rewrite system by which any Clifford circuit can be
reduced to normal form. This yields a presentation of Clifford
operators in terms of generators and relations, shown in
Figure~\ref{fig-relations}.

\paragraph{Related work.}

Van den Nest {\cite{VandenNest}} gave a ``normal'' form for Clifford
circuits, showing that every Clifford circuit can be written as a
single layer of Hadamard gates, sandwiched between two circuits
consisting only of gates that preserve the computational basis ($X$,
$S$, controlled-$X$, and controlled-$Z$ gates). Since it is evident
that such circuits can be efficiently simulated on a classical
computer, this yields a direct proof of the Gottesman-Knill theorem
{\cite{Gottesman-1998}} without relying on the stabilizer formalism.
However, Van den Nest's normal forms are not at all unique, and
therefore cannot be used to derive an equational presentation of
Clifford operators.

Backens {\cite{Backens}} showed completeness of the ZX-calculus, a
graphical language (given by generators and relations) that
generalizes quantum circuits and includes Clifford circuits as a
proper subset. While this work is closely related, it does not yield a
direct equational presentation of Clifford operators. This is because
expressions of the ZX-calculus can denote general linear maps, and not
just unitary ones.

\section{Generators and relations}

Before continuing, it may be worthwhile to clarify what we mean by
``generators and relations''. We do not mean this in the sense of the
usual {\em word problems} studied in group theory, but in the sense of
{\em two-dimensional word problems} appropriate for quantum circuits.
The use of higher-dimensional systems of generators and relations was
pioneered by Burroni~\cite{Burroni-1993}, and was used, for example,
by Lafont to axiomatize various classes of boolean
circuits~\cite{Lafont-2003}.

In a nutshell, higher-dimensional word problems are a generalization
of word problems where one regards not only composition, but also
tensor product as a basic structural operation. We already mentioned
that the Clifford operators form a family of groups $\Clifford(0)$,
$\Clifford(1)$, $\Clifford(2)$, etc. This family is equipped with the
additional structure of a {\em strict spatial monoidal groupoid} (see
{\cite{ML71,Sel2009}}). The abstract definition of strict spatial
monoidal groupoids is not of great importance here; for our purposes,
it simply means the following: the Clifford operators are equipped
with an associative tensor product $\otimes :
\Clifford(n)\times\Clifford(m)\to\Clifford(n+m)$, such that the
identity group element of $\Clifford(0)$ also serves as the left and
right unit for tensor, and satisfying the {\em bifunctorial law}
$(f\otimes I_m)\circ(I_n \otimes g) = (I_n \otimes g)\circ(f\otimes
I_m)$ and the {\em spatial law} $\lambda\otimes I_n =
I_n\otimes\lambda$, where $f\in\Clifford(n)$, $g\in\Clifford(m)$,
$\lambda\in\Clifford(0)$, and $I_n$, $I_m$ are the identity elements
of $\Clifford(n)$ and $\Clifford(m)$, respectively. (From now on, we
omit the subscript, writing $I$ for the identity matrix of any size,
and in fact for the identity element of any group). In circuit
notation:
\[
  \mbox{Bifunctorial law:}~~
  \mp{0.3}{\begin{qcircuit}[scale=0.5]
      \grid{3.5}{0,1}
      \gate{$f$}{1,1}
      \gate{$g$}{2.5,0}
    \end{qcircuit}
  }
  ~=~
  \mp{0.3}{\begin{qcircuit}[scale=0.5]
      \grid{3.5}{0,1}
      \gate{$g$}{1,0}
      \gate{$f$}{2.5,1}
    \end{qcircuit}
  };
  \qquad
  \mbox{Spatial law:}~~
  \mp{0.3}{\begin{qcircuit}[scale=0.5]
      \grid{3.5}{0}
      \gate{$\lambda$}{1.75,1}
    \end{qcircuit}
  }
  ~=~
  \mp{0.3}{\begin{qcircuit}[scale=0.5]
      \grid{3.5}{0}
      \gate{$\lambda$}{1.75,-1}
    \end{qcircuit}
  }.
\]

Thus, the notion of a strict spatial monoidal groupoid already has the
notion of the tensor product ``built in'', along with the fact that
operators on disjoint sets of qubits commute with each other, and that
scalars commute with everything.

Consequently, when we give generators and relations for Clifford
operators {\em as a strict spatial monoidal groupoid}, the
bifunctorial and spatial laws do not need to appear explicitly as part
of the axiomatization. Moreover, unlike group theoretic
axiomatizations, we only need one generator per basic gate, and not
one generator per basic gate per qubit.

As a further illustration of this concept, consider the usual
axiomatization of the braid group $\cB(n)$ on $n$ strands. If this is
axiomatized as a group, one requires $n-1$ generators
$\sigma_1,\ldots,\sigma_{n-1}$ (provided that $n\geq 2$), as well as
$n-2$ instances of the Yang-Baxter equation
$\sigma_i\sigma_{i+1}\sigma_i=\sigma_{i+1}\sigma_i\sigma_{i+1}$
(provided that $n\geq 3$), and $(n-2)(n-3)/2$ instances of
commutativity $\sigma_i\sigma j=\sigma_j\sigma_i$ where $j\geq i+2$
(provided that $n\geq 4$). On the other hand, the axiomatization in
terms of strict spatial monoidal groupoid requires only one generator
$\sigma\in\cB(2)$ and one equation $(\sigma\times
I)(I\times\sigma)(\sigma\times I) = (I\times\sigma)(\sigma\times
I)(I\times\sigma)$, where $I\in\cB(1)$ is the group identity.

\section{Action of the Clifford group on the Pauli group}

Let $X$, $Y$, and $Z$ be the usual Pauli operators
\begin{equation}\label{eqn-pauli}
X = \zmatrix{cc}{0&1\\1&0},\quad
Y = \zmatrix{cc}{0&-i\\i&0},\quad
Z = \zmatrix{cc}{1&0\\0&-1}.
\end{equation}
The {\em Pauli group on $n$ qubits} consists of
$2^n\times2^n$-matrices of the form $\lambda P_1\otimes\ldots\otimes
P_n$, where $\lambda\in\s{\pm1,\pm i}$ and
$P_1,\ldots,P_n\in\s{I,X,Y,Z}$. We write $\Pauli(n)$ for the Pauli
group on $n$ qubits. 

We say that an $n$-qubit operator $U$ is a {\em scalar} if it is a
scalar multiple of the identity operator, i.e., $U=\lambda I$. In this
case, we also write $U=\lambda$ by a mild abuse of notation.

It is well-known that the Clifford group acts on the Pauli group by
conjugation: whenever $C\in\Clifford(n)$ is a Clifford operator and
$P\in\Pauli(n)$ is a Pauli operator, then $C\bullet P := CPC\inv
\in\Pauli(n)$ is another Pauli operator. Moreover, because $C\bullet
(PQ) = (C\bullet P)(C\bullet Q)$, the action of any fixed Clifford
operator $C$ on $\Pauli(n)$ is a group automorphism of
$\Pauli(n)$. Also, since $C\bullet\lambda=\lambda$, this automorphism
fixes scalars. We will show in Proposition~\ref{prop-automorphism}
below that, conversely, any such group automorphism arises from the
action of some Clifford operator.
We have the following well-known properties:

\begin{proposition}\label{prop-clifford-unique}
  Let $C\in\Clifford(n)$. If $C\bullet P = P$ for all $P\in\Pauli(n)$,
  then $C$ is a scalar. 
\end{proposition}

\begin{proof}
  First note that every complex $2\times 2$-matrix can be written in
  the form $aI+bX+cY+dZ$, for complex scalars $a,b,c,d$. It follows
  that the set of Pauli operators spans the set of $2^n\times
  2^n$-operators as a vector space.  By assumption, $CPC\inv = P$ for
  all Pauli operators $P$. It follows that $CMC\inv = M$, hence
  $CM=MC$, for all operators $M$. This implies that $C$ is a scalar.
\end{proof}

\begin{corollary}\label{cor-uniqueness}
  If $C,D$ are two Clifford operators that act identically on the
  Pauli group, then $C,D$ differ only by a global phase, i.e.,
  $C=\omega^iD$ for some $i$.
\end{corollary}

\begin{proof}
  By Proposition~\ref{prop-clifford-unique}, applied to $D\inv C$.
\end{proof}

\begin{proposition}\label{prop-clifford-exists}
  Let $\phi:\Pauli(n)\to\Pauli(n)$ be an automorphism of the Pauli
  group that fixes scalars. Then there exists some Clifford operator
  $C\in\Clifford(n)$ (necessarily unique up to a phase by
  Corollary~\ref{cor-uniqueness}) such that for all $P$, $C\bullet P =
  \phi(P)$.
\end{proposition}

Proposition~\ref{prop-clifford-exists} is an immediate consequence of
Proposition~\ref{prop-automorphism}, which we will prove below. 

\section{Normal forms for Clifford operators}

We follow the usual practice of writing quantum circuits from left to
right, i.e., in the opposite order of the notation for matrix
multiplication. The qubits in a circuit are numbered from top to
bottom. The gates $\omega$, $H$, $S$, and $Z_c$ were defined in
{\eqref{eqn-generators}}, and are respectively called the omega-gate,
Hadamard gate, $S$-gate, and controlled-$Z$ gate. We use the usual
circuit notations for the Hadamard and $S$-gates. Because the
controlled-$Z$ gate is symmetric in the two qubits it acts upon, we
use a symmetric notation for it:
\[ \mbox{Hadamard gate:}~~
\mp{.2}{\begin{qcircuit}[scale=0.5]
    \grid{2}{0}
    \gate{$H$}{1,0}
  \end{qcircuit}
  };
\quad
\mbox{$S$-gate:}~~
\mp{.2}{\begin{qcircuit}[scale=0.5]
    \grid{2}{0}
    \gate{$S$}{1,0}
  \end{qcircuit}
  };
\quad
\mbox{Controlled $Z$-gate:}~~
\mp{0.3}{\begin{qcircuit}[scale=0.5]
    \grid{2}{0,1}
    \controlled{\dotgate}{1,0}{1}
  \end{qcircuit}}
=
\mp{0.4}{\begin{qcircuit}[scale=0.5]
    \grid{2}{0,1}
    \controlled{\gate{$Z$}}{1,0}{1}
  \end{qcircuit}}
=
\mp{0.2}{\begin{qcircuit}[scale=0.5]
    \grid{2}{0,1}
    \controlled{\gate{$Z$}}{1,1}{0}
  \end{qcircuit}
}.
\]
The Pauli operators {\eqref{eqn-pauli}} are also Clifford operators,
and are definable as $X=HSSH$, $Y=HSSHSS\omega^2$, $Z=SS$.

\begin{remark}
  In this paper, we assume that controlled-$Z$ gates, and other binary
  gates, are only applied to {\em adjacent} qubits. This is without
  loss of generality, because gates on non-adjacent qubits can be
  equivalently expressed using swap gates:
  \begin{equation}\label{eqn-adjacent}
    \m{\begin{qcircuit}[scale=0.5]
        \grid{2}{0,1,2}
        \controlled{\dotgate}{1,0}{2}
      \end{qcircuit}
    }
    =
    \m{\begin{qcircuit}[scale=0.5]
        \grid{5}{2}
        \grid{1}{0,1}
        \gridx{2}{3}{0,1}
        \gridx{4}{5}{0,1}
        \draw (1,0) -- (2,1);
        \draw (1,1) -- (2,0);
        \draw (3,0) -- (4,1);
        \draw (3,1) -- (4,0);
        \controlled{\dotgate}{2.5,1}{2}
      \end{qcircuit}~,
    }
  \end{equation}
  and swap gates can be further decomposed as
  \begin{equation}\label{eqn-swap}
    \m{\begin{qcircuit}[scale=0.5]
        \gridx{0}{1}{0,1}
        \gridx{2}{3}{0,1}
        \draw (1,0) -- (2,1);
        \draw (1,1) -- (2,0);
      \end{qcircuit}
    }
    =
    \m{\begin{qcircuit}[scale=0.5]
        \grid{8.5}{0,1}
        \gate{$H$}{1.25,0}
        \gate{$H$}{1.25,1}
        \controlled{\dotgate}{2.5,0}{1}
        \gate{$H$}{3.75,0}
        \gate{$H$}{3.75,1}
        \controlled{\dotgate}{5,0}{1}
        \gate{$H$}{6.25,0}
        \gate{$H$}{6.25,1}
        \controlled{\dotgate}{7.5,0}{1}
      \end{qcircuit}~.
    }
  \end{equation}
  The restriction of gates to adjacent qubits leads to a much cleaner
  presentation of normal forms.
\end{remark}

\begin{figure}
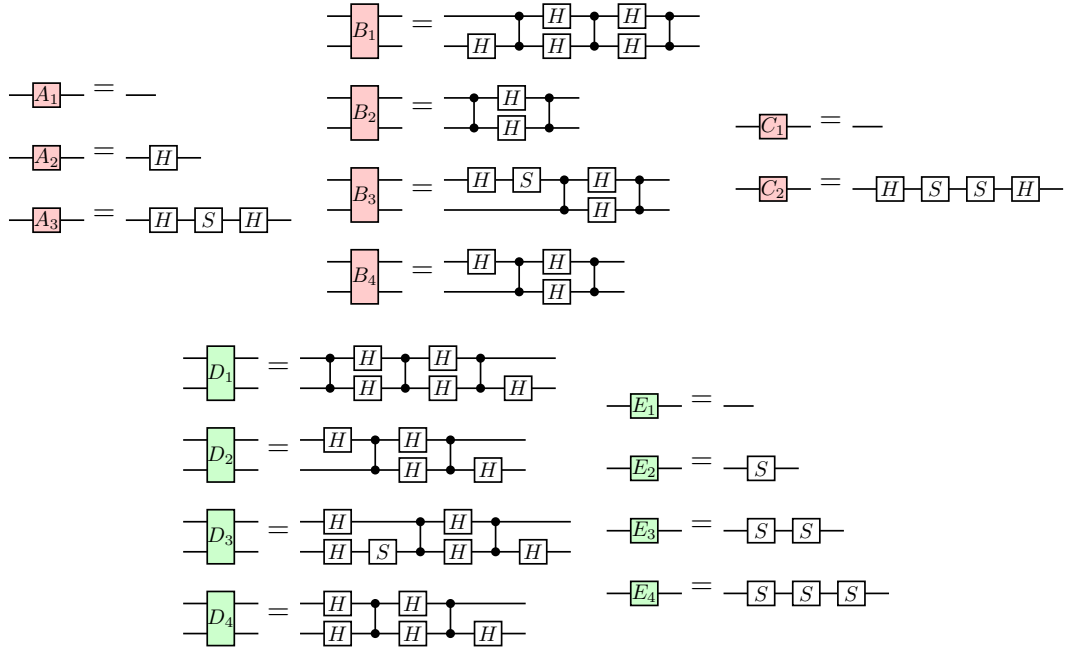

\[ \begin{array}{l@{}l@{~}l}
  \defa
\end{array}
\begin{array}{l@{}l@{~}l} 
  \defb
\end{array}
\begin{array}{l@{}l@{~}l} 
  \defc
\end{array}
\]\[
\begin{array}{l@{}l@{~}l} 
  \defd
\end{array}
\begin{array}{l@{}l@{~}l} 
  \defe
\end{array}
\]
\caption{A convenient gate set for Clifford circuits.}\label{fig-def}
\end{figure}

\begin{definition}
  We define new basic gates $A_i$, $B_j$, $C_k$, $D_{\ell}$, and
  $E_h$, where $i\in\s{1,2,3}$, $j,\ell,h\in\s{1,2,3,4}$, and
  $k\in\s{1,2}$. The meaning of these gates is given in
  Figure~\ref{fig-def}. Note that this gate set is highly redundant;
  for example, $A_1$, $C_1$, and $E_1$ are three different notations
  for the single-qubit identity gate. However, these gates will be
  convenient as building blocks for normal forms. They were chosen for
  their particular actions on Pauli operators, as will be explained in
  Section~\ref{sec-existence-uniqueness} and shown in
  Figure~\ref{fig-pauli}.
\end{definition}

\begin{definition}
  We say that an $n$-qubit circuit is {\em $Z$-normal} if it is of the
  form
  \begin{equation}
  \m{\begin{qcircuit}[scale=0.7]
      \grid{3}{0,1,2,3,4,5,6}
      \widelgate{.75}{^{(n)}}{1.5}{0}{6}
      \wirelabel{$\bvdots$}{0.375,0.2}
      \wirelabel{$\bvdots$}{2.625,0.2}
      \wirelabel{$\bvdots$}{0.375,4.2}
      \wirelabel{$\bvdots$}{2.625,4.2}
    \end{qcircuit}}
  =
  \mp{0.475}{\begin{qcircuit}[scale=0.7]
      \grid{9.8}{0,1,2,3,4,5,6}
      \agate{i}{1}{2}
      \bgate{j_1}{2.5}{2}{3}
      \bgate{j_2}{4}{3}{4}
      \colgate{white, white}{$\bcdots$}{5.5,4}
      \colgate{white, white}{$\bcdots$}{5.5,5}
      \widebgate{0.75}{j_{m-1}}{7.15}{5}{6}
      \cgate{k}{8.8}{6}
      \wirelabel{$\bvdots$}{1,0.2}
      \wirelabel{$\bvdots$}{8.8,0.2}
      \wirelabel{$\bvdots$}{1,4.2}
      \wirelabel{$\bvdots$}{8.8,4.2}
    \end{qcircuit}~,}
  \end{equation}
  for $1\leq m\leq n$. We say that an $n$-qubit circuit is {\em
    $X$-normal} if it is of the form
  \begin{equation}
  \m{\begin{qcircuit}[scale=0.7]
      \grid{3}{0,1,2,3,4,5,6}
      \widemgate{.75}{^{(n)}}{1.5}{0}{6}
      \wirelabel{$\bvdots$}{0.375,1.2}
      \wirelabel{$\bvdots$}{2.625,1.2}
    \end{qcircuit}}
  =
  \m{\begin{qcircuit}[scale=0.7] \grid{11}{0,1,2,3,4,5,6}
      \dgate{\ell_1}{1}{5}{6}
      \dgate{\ell_2}{2.5}{4}{5}
      \dgate{\ell_3}{4}{3}{4}
      \dgate{\ell_4}{5.5}{2}{3}
      \whitegate{$\bcdots$}{7}{2}
      \whitegate{$\bcdots$}{7}{1}
      \widedgate{.75}{\ell_{n-1}}{8.5}{0}{1}
      \egate{h}{10}{0}
      \wirelabel{$\bvdots$}{1.75,1.2}
      \wirelabel{$\bvdots$}{10,1.2}
    \end{qcircuit}~.}
  \end{equation}
  Finally, an $n$-qubit circuit is {\em normal} if it is of the form
  \begin{equation}\label{eqn-normal-form}
  \m{\begin{qcircuit}[scale=0.7]
      \grid{3}{0,1,2,3}
      \widengate{.75}{^{(n)}}{1.5}{0}{3}
      \wirelabel{$\bvdots$}{0.375,1.2}
      \wirelabel{$\bvdots$}{2.625,1.2}
    \end{qcircuit}}
  =
  \m{\begin{qcircuit}[scale=0.7]
      \grid{20.5}{0,1,2,3}
      \wirelabel{$\bvdots$}{0.375,1.2}
      \widelgate{.75}{^{(n)}}{1.5}{0}{3}
      \wirelabel{$\bvdots$}{2.625,1.2}
      \widemgate{.75}{^{(n)}}{3.75}{0}{3}
      \wirelabel{$\bvdots$}{4.875,1.2}
      \widelgate{.75}{^{(n-1)}}{6.0}{1}{3}
      \wirelabel{$\bvdots$}{7.125,1.2}
      \widemgate{.75}{^{(n-1)}}{8.25}{1}{3}
      \wirelabel{$\bvdots$}{9.5,1.2}
      \whitegate{$\bcdots$}{10.25}{3}
      \whitegate{$\bcdots$}{10.25}{2}
      \wirelabel{$\bvdots$}{11.0,1.2}
      \widelgate{.75}{^{(2)}}{12.25}{2}{3}
      \wirelabel{$\bvdots$}{13.375,1.2}
      \widemgate{.75}{^{(2)}}{14.5}{2}{3}
      \wirelabel{$\bvdots$}{15.625,1.2}
      \widelgate{.75}{^{(1)}}{16.75}{3}{3}
      \wirelabel{$\bvdots$}{17.875,1.2}
      \widemgate{.75}{^{(1)}}{19.0}{3}{3}
      \wirelabel{$\bvdots$}{20.125,1.2}
    \end{qcircuit}}\cdot \omega^p,
  \end{equation}
  where $p\in\s{0,1,\ldots,7}$.
\end{definition}

\section{Existence and uniqueness of normal forms}
\label{sec-existence-uniqueness}

When $C$ is a Clifford operator, and $P=P_1\otimes\ldots\otimes P_n$
and $Q=Q_1\otimes\ldots\otimes Q_n$ are Pauli operators, we 
schematically write 
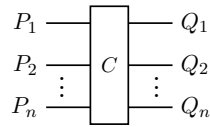
\[ \m{\begin{qcircuit}[scale=0.7]
    \grid{3}{0,1,2}
    \leftlabel{$P_n$}{0,0}
    \wirelabel{$\bvdots$}{0.375,0.2}
    \leftlabel{$P_2$}{0,1}
    \leftlabel{$P_1$}{0,2}
    \biggate{$C$}{1.5,0}{1.5,2}
    \rightlabel{$Q_n$}{3,0}
    \wirelabel{$\bvdots$}{2.625,0.2}
    \rightlabel{$Q_2$}{3,1}
    \rightlabel{$Q_1$}{3,2}
  \end{qcircuit}}
\]
to indicate that $C\bullet P = Q$. With that convention,
Figure~\ref{fig-pauli} shows the action of the operators $A_i$, $B_j$,
$C_k$, $D_{\ell}$, and $E_h$ on selected Pauli operators.

\begin{figure}
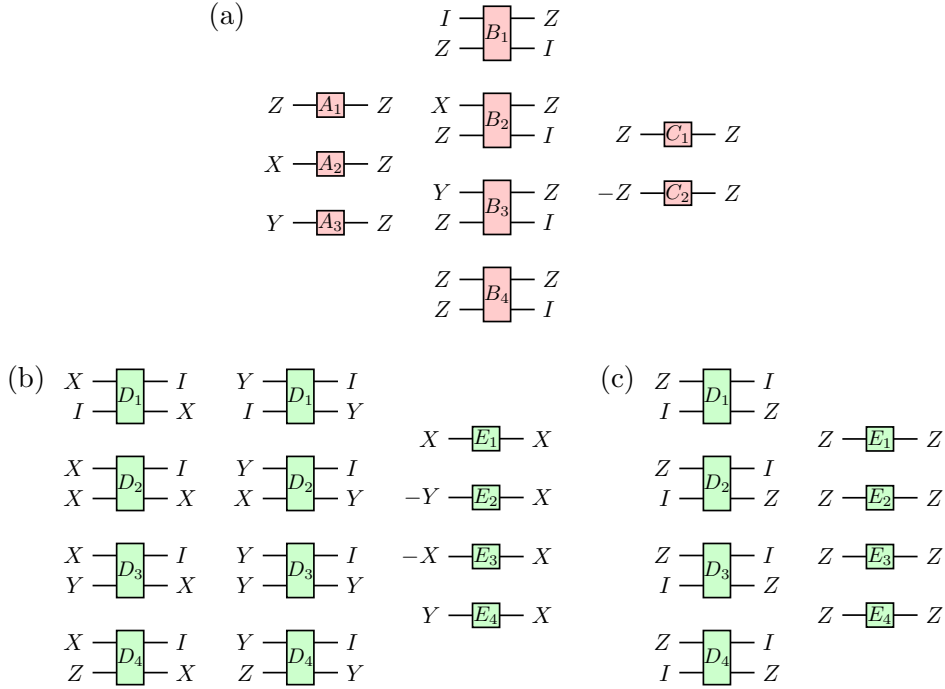

\[ \raisebox{.8in}{(a)}\m{$
\begin{array}{r}
  \paulia
\end{array}
\begin{array}{r}
  \paulib
\end{array}
\begin{array}{r}
  \paulic
\end{array}
$}
\]
\[ \raisebox{.8in}{(b)}\m{$
\begin{array}{r}
  \paulifx
\end{array}
\begin{array}{r}
  \paulify
\end{array}
\begin{array}{r}
  \pauliex
\end{array}
$}\quad
\raisebox{.8in}{(c)}\m{$
\begin{array}{r}
  \paulifz
\end{array}
\begin{array}{r}
  \pauliez
\end{array}
$}
\]
\caption{The action of the operators $A_i$, $B_j$, $C_k$, $D_{\ell}$, and
  $E_h$ on certain Pauli operators.}\label{fig-pauli}
\end{figure}

\begin{lemma}\label{lem-l}
  Let $n\geq 1$, and let $P$ be an $n$-qubit Pauli operator satisfying
  $P^2=I$ and $P\neq \pm I$. Then there exists a unique $Z$-normal
  circuit $L$ such that
  \[ L\bullet P = Z\tensor I\tensor\ldots\tensor I.
  \]
\end{lemma}

\begin{proof}
  As a Pauli operator, $P$ is of the form $P=\lambda P_1\tensor
  P_2\tensor\ldots\tensor P_n$, where $P_1,\ldots,P_n\in
  \s{I,X,Y,Z}$. The requirement $P^2=I$ ensures that $\lambda$ is
  $\pm1$ (and not $\pm i$). Moreover, since $P\neq\pm I$, there must
  be some $m$ such that $P_m\neq I$; let $m$ be the largest such
  index. The claim then follows from the properties in
  Figure~\ref{fig-pauli}(a), with reference to the following diagram:
  \[
  \mp{0}{\begin{qcircuit}[scale=0.7]
      \gridx{-1.5}{15.8}{0,1,2,3,4,5,6}
      \leftlabel{$P_1$}{-1.5,6}
      \leftlabel{$P_2$}{-1.5,5}
      \leftlabel{$P_{m-2}$}{-1.5,4}
      \leftlabel{$P_{m-1}$}{-1.5,3}
      \leftlabel{$\pm P_{m}$}{-1.5,2}
      \leftlabel{$I$}{-1.5,1}
      \leftlabel{$I$}{-1.5,0}
      \agate{i}{-0.5}{2}
      \whitegate{$\pm Z$}{1}{2}
      \bgate{j_1}{2.5}{2}{3}
      \whitegate{$\pm Z$}{4}{3}
      \bgate{j_2}{5.5}{3}{4}
      \whitegate{$\pm Z$}{7}{4}
      \whitegate{$\bcdots$}{8.5}{5}
      \whitegate{$\bcdots$}{8.5}{4}
      \whitegate{$\pm Z$}{10}{5}
      \widebgate{0.75}{j_{m-1}}{11.65}{5}{6}
      \whitegate{$\pm Z$}{13.3}{6}
      \cgate{k}{14.8}{6}
      \wirelabel{$\bvdots$}{-0.5,0.2}
      \wirelabel{$\bvdots$}{-0.5,4.2}
      \wirelabel{$\bvdots$}{14.8,0.2}
      \wirelabel{$\bvdots$}{14.8,4.2}
      \rightlabel{$Z$}{15.8,6} 
      \rightlabel{$I$}{15.8,5} 
      \rightlabel{$I$}{15.8,4} 
      \rightlabel{$I$}{15.8,3} 
      \rightlabel{$I$}{15.8,2} 
      \rightlabel{$I$}{15.8,1} 
      \rightlabel{$I.$}{15.8,0} 
    \end{qcircuit}}
  \]
  In particular, note that it follows from the properties in
  Figure~\ref{fig-pauli}(a) that there exists a unique $A_i$ sending
  $\pm P_m$ to $\pm Z$; for each $p=1,\ldots,m-1$, there exists a
  unique $B_{j_p}$ sending $P_{m-p}\tensor \pm Z$ to $\pm Z\tensor I$;
  and there exists a unique $C_k$ sending $\pm Z$ to $Z$.
\end{proof}

\begin{lemma}\label{lem-m}
  Let $n\geq 1$, and let $Q$ be an $n$-qubit Pauli operator satisfying
  $Q^2=I$ such that $Q$ anticommutes with $Z\tensor
  I\tensor\ldots\tensor I$. Then there exists a unique $X$-normal
  circuit $M$ such that
  \[ M\bullet Q = I\tensor \ldots\tensor I\tensor X. 
  \]
\end{lemma}

\begin{proof}
  As before, since $Q^2=I$, we know that $Q$ is of the form $Q = \pm
  Q_1\tensor Q_2\tensor\ldots\tensor Q_n$. Moreover, since $Q$
  anticommutes with $Z\tensor I\tensor\ldots\tensor I$, we have that
  $Q_1 \in \s{X,Y}$. The claim then follows from the properties in
  Figure~\ref{fig-pauli}(b), with reference to the following diagram:
  \[
  \mp{0}{\begin{qcircuit}[scale=0.7]
      \gridx{1.5}{18.8}{1,2,3,4,5,6}
      \leftlabel{$\pm Q_1$}{1.5,6}
      \leftlabel{$Q_2$}{1.5,5}
      \leftlabel{$Q_3$}{1.5,4}
      \leftlabel{$Q_4$}{1.5,3}
      \leftlabel{$Q_{n-1}$}{1.5,2}
      \leftlabel{$Q_n$}{1.5,1}
      \dgate{\ell_1}{2.5}{5}{6}
      \widewhitegate{.55}{$\pm Q_1$}{4}{5}
      \dgate{\ell_2}{5.5}{4}{5}
      \widewhitegate{.55}{$\pm Q_1$}{7}{4}
      \dgate{\ell_3}{8.5}{3}{4}
      \widewhitegate{.55}{$\pm Q_1$}{10}{3}
      \whitegate{$\bcdots$}{11.5}{3}
      \whitegate{$\bcdots$}{11.5}{2}
      \widewhitegate{.55}{$\pm Q_1$}{13}{2}
      \widedgate{.75}{\ell_{n-1}}{14.65}{1}{2}
      \widewhitegate{.55}{$\pm Q_1$}{16.3}{1}
      \egate{h}{17.8}{1}
      \wirelabel{$\bvdots$}{2.5,2.2}
      \wirelabel{$\bvdots$}{17.8,2.2}
      \rightlabel{$I$}{18.8,6} 
      \rightlabel{$I$}{18.8,5} 
      \rightlabel{$I$}{18.8,4} 
      \rightlabel{$I$}{18.8,3} 
      \rightlabel{$I$}{18.8,2} 
      \rightlabel{$X.$}{18.8,1} 
    \end{qcircuit}}
  \]
\end{proof}

\begin{lemma}\label{lem-mz}
  Every $X$-normal circuit $M$ satisfies
  \[ M\bullet (Z\tensor I\tensor\ldots\tensor I) = I\tensor\ldots\tensor I\tensor Z.
  \]
\end{lemma}

\begin{proof}
  By using the properties from Figure~\ref{fig-pauli}(c), with
  reference to the following diagram:
  \[
  \mp{0}{\begin{qcircuit}[scale=0.7]
      \gridx{1.5}{18.8}{1,2,3,4,5,6}
      \leftlabel{$Z$}{1.5,6}
      \leftlabel{$I$}{1.5,5}
      \leftlabel{$I$}{1.5,4}
      \leftlabel{$I$}{1.5,3}
      \leftlabel{$I$}{1.5,2}
      \leftlabel{$I$}{1.5,1}
      \dgate{\ell_1}{2.5}{5}{6}
      \widewhitegate{.55}{$Z$}{4}{5}
      \dgate{\ell_2}{5.5}{4}{5}
      \widewhitegate{.55}{$Z$}{7}{4}
      \dgate{\ell_3}{8.5}{3}{4}
      \widewhitegate{.55}{$Z$}{10}{3}
      \whitegate{$\bcdots$}{11.5}{3}
      \whitegate{$\bcdots$}{11.5}{2}
      \widewhitegate{.55}{$Z$}{13}{2}
      \widedgate{.75}{\ell_{n-1}}{14.65}{1}{2}
      \widewhitegate{.55}{$Z$}{16.3}{1}
      \egate{h}{17.8}{1}
      \wirelabel{$\bvdots$}{2.5,2.2}
      \wirelabel{$\bvdots$}{17.8,2.2}
      \rightlabel{$I$}{18.8,6} 
      \rightlabel{$I$}{18.8,5} 
      \rightlabel{$I$}{18.8,4} 
      \rightlabel{$I$}{18.8,3} 
      \rightlabel{$I$}{18.8,2} 
      \rightlabel{$Z.$}{18.8,1} 
    \end{qcircuit}}
  \]
\end{proof}

\begin{lemma}\label{lem-pq}
  Let $n\geq 1$, and let $P$ and $Q$ be $n$-qubit Pauli operators such
  that $P^2=Q^2=I$, and such that $P$, $Q$ anticommute. Then there
  exist unique circuits $M$ and $L$, where $M$ is $X$-normal and $L$
  is $Z$-normal, such that
  \[ ML \bullet P = I\tensor\ldots\tensor I\tensor Z\]
  and
  \[ ML \bullet Q = I\tensor\ldots\tensor I\tensor X.\]
\end{lemma}

\begin{proof}
  By Lemma~\ref{lem-l}, there exists a $Z$-normal circuit $L$ such
  that $L\bullet P = Z\tensor I\tensor\ldots\tensor I$. Since $P$ and
  $Q$ anticommute, so do $L\bullet P$ and $L\bullet Q$. Therefore, by
  Lemma~\ref{lem-m}, there exists an $X$-normal circuit $M$ such that
  $M\bullet (L\bullet Q) = I\tensor\ldots\tensor I\tensor X$, i.e.,
  $ML\bullet Q = I\tensor\ldots\tensor I\tensor X$. By
  Lemma~\ref{lem-mz}, we also have $ML\bullet P = M\bullet (L\bullet
  P) = M\bullet (Z\tensor I\tensor\ldots\tensor I) = I\tensor\ldots\tensor I\tensor Z$. This proves the existence of $M$ and $L$.
  For uniqueness, assume that $M'$ and $L'$ are another pair of
  operators satisfying the conditions of the lemma. From
  $M'\bullet(L'\bullet P) = I\tensor\ldots\tensor I\tensor Z$, we have
  $L'\bullet P=Z\tensor I\tensor\ldots\tensor I$ by
  Lemma~\ref{lem-mz}. Therefore, $L=L'$ by uniqueness of $L$ in
  Lemma~\ref{lem-l}. But then $M'\bullet(L\bullet Q) =
  M\bullet(L\bullet Q)=X\tensor I\tensor\ldots\tensor I$, so that
  $M=M'$ by uniqueness of $M$ in Lemma~\ref{lem-m}.
\end{proof}

\begin{proposition}\label{prop-automorphism}
  Let $\phi:\Pauli(n)\to\Pauli(n)$ be an automorphism of the Pauli
  group such that $\phi$ fixes scalars. Then there exists a Clifford
  circuit $C$ in normal form {\eqref{eqn-normal-form}} such that for
  all $P$, $C\bullet P = \phi(P)$. Moreover, the normal form $C$ is
  unique up to the scalar $\omega^p$.
\end{proposition}

\begin{proof}
  By induction on $n$. For $n=0$, all Pauli operators are scalars, so
  $\phi$ is the identity; we can set $C=1=\omega^0$; uniqueness up to
  a scalar follows because for $n=0$, all Clifford operators are scalars.

  For the induction step, suppose the claim is true for $n-1$. We
  first prove existence.  Let $P=\phi\inv(I\tensor\ldots\tensor
  I\tensor Z)$ and $Q=\phi\inv(I\tensor\ldots\tensor I\tensor
  X)$. Note that $I\tensor\ldots\tensor I\tensor Z$ and $I\tensor\ldots\tensor I\tensor X$ each square to the identity and anticommute
  with each other; since $\phi$ is an automorphism, the same is true
  for $P$ and $Q$. By Lemma~\ref{lem-pq}, there exists a circuit $ML$,
  where $M$ is $X$-normal and $L$ is $Z$-normal, such that
  \begin{equation}\label{eqn-mlp}
    ML\bullet P = I\tensor\ldots\tensor I\tensor Z =
    \phi(P)
  \end{equation}
  and 
  \begin{equation}\label{eqn-mlq}
    ML\bullet Q = I\tensor\ldots\tensor I\tensor X = \phi(Q).
  \end{equation}
  Now define a new automorphism $\phi'$ by $\phi'(U) =
  \phi((ML)\inv\bullet U)$ for all $U\in\Pauli(n)$. Note that $\phi'$
  fixes scalars; also, by {\eqref{eqn-mlp}} and {\eqref{eqn-mlq}},
  $I\tensor\ldots\tensor I\tensor Z$ and $I\tensor\ldots\tensor
  I\tensor X$ are fixed points of $\phi'$. Let $R$ be an $n-1$-qubit
  Pauli operator, and consider $\phi'(R\tensor I)$. Since $R\tensor I$
  commutes with both $I\tensor\ldots\tensor I\tensor Z$ and
  $I\tensor\ldots\tensor I\tensor X$, the same is true for
  $\phi'(R\tensor I)$; therefore, $\phi'(R\tensor I) = S\tensor I$,
  for some $S\in\Pauli(n-1)$. It follows that there exists an
  automorphism $\phi'':\Pauli(n-1)\to\Pauli(n-1)$ such that
  $\phi'(R\tensor I)=\phi''(R)\tensor I$, for all $R\in\Pauli(n-1)$.
  Together with the fact that $I\tensor\ldots\tensor I\tensor Z$ and
  $I\tensor\ldots\tensor I\tensor X$ are fixed points of $\phi'$, this
  implies that $\phi'=\phi''\tensor I$.  By induction hypothesis,
  there exists a normal $n-1$-qubit Clifford circuit $C'$ such that
  for all $R$, $C'\bullet R=\phi''(R)$. Let $C=(C'\tensor I)ML$. Then
  for all $U$, we have
  \[
  C\bullet U = (C'\tensor I)ML\bullet U =
  (C'\tensor I)\bullet(\phi'^{-1}(\phi(U)))
  = (C'\tensor I)\bullet (\phi''^{-1}\tensor I)(\phi(U)) = \phi(U).
  \]
  This proves the existence of the normal form $C$. 

  For uniqueness, suppose that $D$ is another normal form Clifford
  circuit with $D\bullet U = \phi(U)$ for all Pauli $U$. By definition
  of normal forms, $D$ is of the form $(D'\tensor I)M'L'$, where $M'$
  is $X$-normal, $L'$ is $Z$-normal, and $D'$ is a normal Clifford
  circuit of $n-1$ qubits. Then $D\bullet P = \phi(P) =
  I\tensor\ldots\tensor I\tensor Z$. It follows that $M'L'\bullet P =
  (D'\tensor I)\inv\bullet (I\tensor\ldots\tensor I\tensor Z)=
  I\tensor\ldots\tensor I\tensor Z$, and by a similar argument,
  $M'L'\bullet Q=I\tensor\ldots\tensor I\tensor X$. From the
  uniqueness of Lemma~\ref{lem-pq}, we get $M'=M$ and $L'=L$. Then by
  uniqueness in the induction hypothesis, $D'$ and $C'$ are equal up
  to a scalar, and hence the same is true for $D$ and $C$.
\end{proof}

\begin{proof}[Proof of Proposition~\ref{prop-clifford-exists}]
  Proposition~\ref{prop-clifford-exists} is an immediate consequence
  of Proposition~\ref{prop-automorphism}.
\end{proof}

\begin{corollary}
  The Clifford group on $n$ qubits has exactly
  \[  |\Clifford(n)| = 8\cdot\prod_{i=1}^{n} 2(4^i-1)4^i
  \]
  elements.
\end{corollary}

\begin{proof}
  From the definition of $Z$-normal circuits, there are exactly
  \[ \sum_{m=1}^{n} 3\cdot 4^{m-1}\cdot 2 = 2 (4^n - 1)
  \]
  circuits of the form $L^{(n)}$. Moreover, there are $4^n$ circuits
  of the form $M^{(n)}$, hence $2(4^n-1)4^n$ circuits of the form
  $M^{(n)}L^{(n)}$. Because there are exactly 8 scalars, it follows
  that there are
  \[ 8\cdot\prod_{i=1}^{n} 2(4^i-1)4^i
  \]
  circuits of the form $N^{(n)}$ shown in {\eqref{eqn-normal-form}}.
  By Proposition~\ref{prop-automorphism}, these are in one-to-one
  correspondence with the elements of the $n$-qubit Clifford group.
\end{proof}

\section{Normalization via rewrite rules}

In this section, we will describe an explicit procedure for converting
any given Clifford circuit to normal form by using only a finite
number of equations. This yields a presentation of the Clifford group
by generators and relations.

\begin{definition}
  Consider an $n$-qubit Clifford circuit in normal form:
  \begin{equation}\label{eqn-dirty-nf}
    \mp{0.475}{\scalebox{0.9}{\begin{qcircuit}[scale=0.78]
        \gridx{-0.5}{24.1}{0,1,2,3,4,5,6}
        \alabel{0,6}
        \alabel{0,5}
        \alabel{0,4}
        \alabel{0,3}
        \alabel{0,2}
        \alabel{0,1}
        \alabel{0,0}
        \agate{i}{1}{2}
        \blabel{1.75,2}
        \bgate{j_1}{2.5}{2}{3} 
        \blabel{3.25,3}
        \bgate{j_2}{4}{3}{4}
        \blabel{4.75,4}
        \colgate{white, white}{$\bcdots$}{5.5,4}
        \colgate{white, white}{$\bcdots$}{5.5,5}
        \blabel{6.25,5}
        \widebgate{0.75}{j_{m-1}}{7.3}{5}{6}
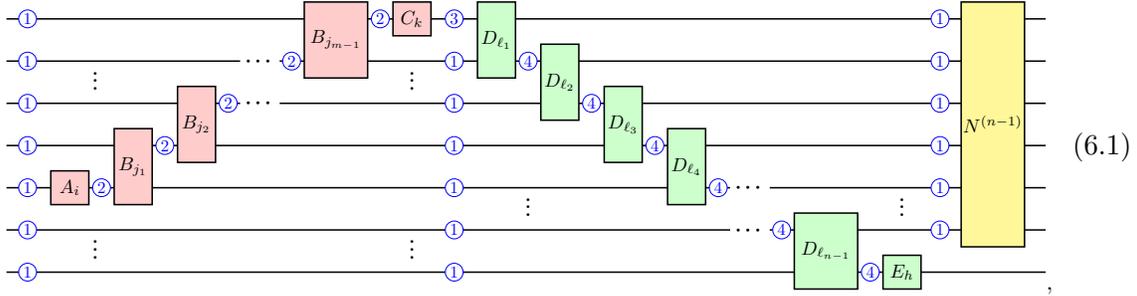
        \begin{scope}[xshift=0.6cm]
          \blabel{7.75,6}
          \cgate{k}{8.5}{6}
          \wirelabel{$\bvdots$}{1,0.3}
          \wirelabel{$\bvdots$}{8.5,0.3}
          \wirelabel{$\bvdots$}{1,4.3}
          \wirelabel{$\bvdots$}{8.5,4.3}
          \clabel{9.5,6}
          \alabel{9.5,5}
          \alabel{9.5,4}
          \alabel{9.5,3}
          \alabel{9.5,2}
          \alabel{9.5,1}
          \alabel{9.5,0}
        \end{scope}
        \begin{scope}[xshift=10.1cm]
          \dgate{\ell_1}{1}{5}{6}
          \dlabel{1.75,5}
          \dgate{\ell_2}{2.5}{4}{5}
          \dlabel{3.25,4}
          \dgate{\ell_3}{4}{3}{4}
          \dlabel{4.75,3}
          \dgate{\ell_4}{5.5}{2}{3}
          \dlabel{6.25,2}
          \whitegate{$\bcdots$}{7}{2}
          \whitegate{$\bcdots$}{7}{1}
          \dlabel{7.75,1}
          \widedgate{.75}{\ell_{n-1}}{8.8}{0}{1}
          \dlabel{9.85,0}
          \egate{h}{10.6}{0}
          \wirelabel{$\bvdots$}{1.75,1.3}
          \wirelabel{$\bvdots$}{10.6,1.3}
        \end{scope}
        \begin{scope}[xshift=21.85cm]
          \alabel{-0.25,6}
          \alabel{-0.25,5}
          \alabel{-0.25,4}
          \alabel{-0.25,3}
          \alabel{-0.25,2}
          \alabel{-0.25,1}
          \widengate{.75}{^{(n-1)}}{1}{1}{6}
        \end{scope}
      \end{qcircuit}},}
  \end{equation}
  We say that a circuit is in {\em dirty normal form} if it is of the
  form {\eqref{eqn-dirty-nf}}, except that the circuit may contain
  some additional gates, subject to the following rules:
  \begin{itemize}
  \item $H$-gates may be added to any wire labelled $\circled{1}$;
  \item $S$-gates may be added to any wire labelled $\circled{1}$,
    $\circled{2}$, $\circled{3}$, or $\circled{4}$;
  \item $X$-gates may be added to any wire labelled $\circled{2}$;
  \item Controlled $Z$-gates may be added to any pair of adjacent
    wires, provided that the top wire is labelled $\circled{1}$,
    $\circled{2}$, or $\circled{3}$, and the bottom wire is labelled
    $\circled{1}$;
  \end{itemize}
  We recursively assume that $N^{(n-1)}$ is in dirty normal form as
  well.
\end{definition}

\begin{lemma}\label{lem-dirty-normal-form}
  Any dirty normal form can be converted to an equivalent normal form
  by using the equations in
  Figures~\ref{fig-rewrite-1}--\ref{fig-rewrite-5}, in the
  left-to-right direction, a finite number of times.
\end{lemma}

\begin{proof}
  By inspection. Let us call the gates of type $A$--$E$ {\em ``clean''},
  and the gates $H$, $S$, $X$, and $Z_c$ {\em ``dirty''}. Given two
  gate occurrences $F$ and $G$ in a circuit, we say that $F$ is {\em
    immediately before} $G$ if one of the outputs of $F$ is connected
  to one of the inputs of $G$. We say that $F$ is {\em before} $G$ if
  there exists a sequence of gates, starting with $F$ and ending with
  $G$, such that each is immediately before the next one.

  By the definition of dirty normal forms, every dirty gate occurs
  before some clean gate. Therefore, as long as there is at least one
  dirty gate in the circuit, some dirty gate (for example, the
  rightmost one) must occur {\em immediately} before a clean gate.
  The left-hand sides of the equations in
  Figures~\ref{fig-rewrite-1}--\ref{fig-rewrite-5} cover all possible
  cases of a dirty gate occurring immediately before a clean gate. So
  as long as there are dirty gates left, one of the equations can
  always be applied. Moreover, a straightforward but tedious
  inspection of the equations in
  Figures~\ref{fig-rewrite-1}--\ref{fig-rewrite-5} shows that the
  left-to-right application of each equation to a dirty normal form
  yields another dirty normal form.

  It remains to be shown that this procedure terminates. To this end,
  we associate to each dirty normal form a sequence $\vec v$ of
  natural numbers as follows.  Suppose the dirty normal form has $r$
  clean gates, which have been numbered $1,\ldots,r$ from left to
  right in the order in which they appear in
  {\eqref{eqn-dirty-nf}}. Then let $\vec v = (v_1,\ldots,v_r)$, where
  $v_i$ is the number of dirty gates before the $i$th clean gate. It
  is easy to see that, with the exception of the equation
  $\omega^8=1$, each left-to-right application of an equation from
  Figures~\ref{fig-rewrite-1}--\ref{fig-rewrite-5} decreases the
  sequence $\vec v$ in the lexicographic ordering.  Although the
  length $r$ of $\vec v$ is not constant, it is bounded by $n^2$, and
  since the set of all such sequences is well-ordered, it follows that
  the procedure terminates in a finite number of steps.
\end{proof}

\begin{proposition}\label{prop-presentation}
  Consider a Clifford circuit expressed in terms of the generators
  $H$, $S$, and controlled-$Z$ gates on adjacent qubits. Any such
  circuit can be converted to its equivalent normal form by finitely
  many uses of the equations in
  Figures~\ref{fig-rewrite-1}--\ref{fig-rewrite-5}, together with the
  equations
  \begin{eqnarray}
    \m{\begin{qcircuit}[scale=0.5]
        \grid{2}{0}
      \end{qcircuit}}
    &=&
    \m{\begin{qcircuit}[scale=0.5]
        \grid{5}{0}
        \agate{1}{1}{0}
        \cgate{1}{2.5}{0}
        \egate{1}{4}{0}
      \end{qcircuit},}
    \label{eqn-ACE}
    \\\nonumber\\[-1ex]
    \mp{.66}{\begin{qcircuit}[scale=0.5]
        \grid{2}{0,1}
        \cgate{1}{1}{0}
      \end{qcircuit}}
    &=&
    \m{\begin{qcircuit}[scale=0.5]
        \grid{5}{0,1}
        \bgate{1}{1}{0}{1}
        \cgate{1}{2.5}{1}
        \dgate{1}{4}{0}{1}
      \end{qcircuit}.}
    \label{eqn-BCD}
  \end{eqnarray}
\end{proposition}

\begin{proof}
  First, note that the normal form of the identity operator on $n$
  qubit is
  \begin{equation}\label{eqn-identity-nf}
  \m{\begin{qcircuit}[scale=0.7]
      \grid{3}{2,3,4,5,6}
      \wideidgate{.75}{^{(n)}}{1.5}{2}{6}
      \wirelabel{$\bvdots$}{0.375,4.2}
      \wirelabel{$\bvdots$}{2.625,4.2}
    \end{qcircuit}}
  ~=~
    \mp{0.475}{\begin{qcircuit}[scale=0.7]
        \gridx{-0.5}{19.5}{2,3,4,5,6}
        \agate{1}{1}{2}
        \bgate{1}{2.5}{2}{3} 
        \bgate{1}{4}{3}{4}
        \colgate{white, white}{$\bcdots$}{5.5,4}
        \colgate{white, white}{$\bcdots$}{5.5,5}
        \bgate{1}{7}{5}{6}
        \cgate{1}{8.5}{6}
        \wirelabel{$\bvdots$}{1,4.2}
        \wirelabel{$\bvdots$}{8.5,4.2}
        \begin{scope}[xshift=9cm]
          \dgate{1}{1}{5}{6}
          \whitegate{$\bcdots$}{2.5}{5}
          \whitegate{$\bcdots$}{2.5}{4}
          \dgate{1}{4}{3}{4}
          \dgate{1}{5.5}{2}{3}
          \egate{1}{7}{2}
          \wirelabel{$\bvdots$}{7,4.2}
        \end{scope}
        \begin{scope}[xshift=17cm]
          \wideidgate{.75}{^{(n-1)}}{1}{3}{6}
        \end{scope}
      \end{qcircuit}.}
  \end{equation}
  The identity circuit can be converted to this normal form by a
  finite number of applications of the equations {\eqref{eqn-ACE}} and
  {\eqref{eqn-BCD}}. By appending the normal form of the identity
  operator {\eqref{eqn-identity-nf}} to the given Clifford circuit, we
  obtain a dirty normal form, which can then be converted to a normal
  form by Lemma~\ref{lem-dirty-normal-form}.
\end{proof}

\section{An equational presentation of the Clifford groupoid}

As an immediate consequence of Proposition~\ref{prop-presentation}, we
know that the equations in
Figures~\ref{fig-rewrite-1}--\ref{fig-rewrite-5}, together with the
equations {\eqref{eqn-ACE}} and {\eqref{eqn-BCD}}, and the defining
equations in Figure~\ref{fig-def}, form a presentation of the Clifford
groupoid by generators and relations. However, this formulation uses a
large number of generators (namely all the gates of type $A$--$E$, as
well as $\omega$, $S$, $H$, $X$, and $Z_c$), and also a very large
number of equations. Naturally, it would be desirable to find a much
smaller set of generators and relations. This is done in the following
proposition.

\begin{proposition}
  The Clifford groupoid is presented, as a strict spatial monoidal
  groupoid, by the generators shown in {\eqref{eqn-generators}} and
  the relations shown in Figure~\ref{fig-relations}.
\end{proposition}

\begin{proof}
  First, we can use the defining equations of Figure~\ref{fig-def} to
  eliminate the generators of type $A$--$E$ from {\eqref{eqn-ACE}},
  {\eqref{eqn-BCD}}, and from the equations in
  Figures~\ref{fig-rewrite-1}--\ref{fig-rewrite-5}; we can also use
  the equation $X=HSSH$ to eliminate the gate $X$. This leaves only
  the generators $\omega$, $S$, $H$, and $Z_c$. It is then a tedious,
  but finite exercise to verify that {\eqref{eqn-ACE}},
  {\eqref{eqn-BCD}}, and each of the equations in
  Figures~\ref{fig-rewrite-1}--\ref{fig-rewrite-5} is a consequence of
  the equations from Figure~\ref{fig-relations}. A full, mostly
  machine-generated proof is available as a supplement to this paper
  {\cite{clifford-supplement}}.
\end{proof}

It is not currently known to the author whether the equations in
Figure~\ref{fig-relations} are independent.

\begin{remark}
  Proving an equation about $2$-qubit circuits never requires any of
  the $3$-qubit axioms to be used; therefore, the equations of
  Figure~\ref{fig-relations}(a)--(c) are sufficient to present the
  $2$-qubit Clifford operators in terms of generators and relations.
  Similarly, the equations of Figure~\ref{fig-relations}(a)--(b) are
  sufficient to present the $1$-qubit Clifford operators, and,
  trivially, the relation of Figure~\ref{fig-relations}(a) is
  sufficient to present the $0$-qubit Clifford operators. It is also
  interesting to note that there are no special axioms involving $4$ or
  more qubits; all equational properties of $n$-qubit Clifford circuits
  follow from axioms involving at most $3$ qubits.
\end{remark}

\begin{figure}
\[ \omega^8 = 1\hspace{0.5in}
\]
\begin{minipage}{0.5\textwidth}
\[ \begin{array}{l@{}l@{~}l}
  \commHA
\end{array}
\]
\end{minipage}%
\begin{minipage}{0.5\textwidth}
\[\begin{array}{l@{}l@{~}l}
  \commSA
  \end{array}
\]
\end{minipage}
\vspace{1ex}
\[\begin{array}{r@{}l@{~}l}
  \commZZAI
  \\
  \commZZIABB
\end{array}
\]
\caption{Rewrite rules for normal forms, part I}\label{fig-rewrite-1}
\end{figure}
\begin{figure}
\begin{minipage}{0.5\textwidth}
\[ \begin{array}{l@{}l@{~}l}
  \commHIB
  \\
  \commISB
\end{array}
\]
\end{minipage}%
\begin{minipage}{0.5\textwidth}
\[\begin{array}{l@{}l@{~}l}
  \commSIB
  \\
  \commIXB
\end{array}
\]
\end{minipage}
\vspace{1ex}
\[
\begin{array}{l@{}l@{~}l}
  \commIZZBBI
\end{array}
\]
\begin{minipage}{0.5\textwidth}
\[ \begin{array}{l@{}l@{~}l}
  \commXC
  \\
  \commSC
\end{array}
\]
\end{minipage}%
\begin{minipage}{0.5\textwidth}
\[ \begin{array}{l@{}l@{~}l}
  \commZZCI
\end{array}
\]
\end{minipage}
\caption{Rewrite rules for normal forms, part II}\label{fig-rewrite-2}
\end{figure}
\begin{figure}
\[
\scalebox{0.87}{$
  \begin{array}{l@{}l@{~}l}
    \commZZIIBBBBI
  \end{array}
  $}
\]
\caption{Rewrite rules for normal forms, part III}\label{fig-rewrite-3}
\end{figure}
\begin{figure}
\begin{minipage}{0.5\textwidth}
\[
\begin{array}{l@{}l@{~}l}
  \altIHDD
  \\
  \altSIDD
\end{array}
\]
\end{minipage}%
\begin{minipage}{0.5\textwidth}
\[
\begin{array}{l@{}l@{~}l}
  \altISDD
  \\
  \altZZDD
\end{array}
\]
\end{minipage}
\[ \begin{array}{l@{}l@{~}l}
  \altSE
\end{array}
\]
\caption{Rewrite rules for normal forms, part IV}\label{fig-rewrite-4}
\end{figure}
\begin{figure}
\[
\scalebox{0.87}{$
  \begin{array}{l@{}l@{~}l}
    \altIZZDDIIDD
  \end{array}
  $}
\]
\caption{Rewrite rules for normal forms, part V}\label{fig-rewrite-5}
\end{figure}

\begin{figure}
  \newcounter{mytmpcounter}
  \setcounter{mytmpcounter}{\value{equation}}
  \setcounter{equation}{0}
  \makeatletter
  \renewcommand{\theequation}{C\@arabic\c@equation}
  \makeatother
  \flushleft

  (a) Equations for $n\geq 0$:
  \begin{equation}
    \omega^8 = 1
  \end{equation}
  (b) Equations for $n\geq 1$:
  \begin{eqnarray}
    H^2 &=& 1\\
    S^4 &=& 1\\
    SHSHSH &=& \omega
  \end{eqnarray}
  (c) Equations for $n\geq 2$:
  \begin{eqnarray}
    \m{\begin{qcircuit}[scale=0.5]
        \grid{3}{0,1}
        \controlled{\dotgate}{1,0}{1}
        \controlled{\dotgate}{2,0}{1}
      \end{qcircuit}
    } 
    &=& 
    \m{\begin{qcircuit}[scale=0.5]
        \grid{2}{0,1}
      \end{qcircuit}
    }
    \\\nonumber\\[0ex]
    \m{\begin{qcircuit}[scale=0.5]
        \grid{3.5}{0,1}
        \gate{$S$}{1.25,1}
        \controlled{\dotgate}{2.5,0}{1}
      \end{qcircuit}
    } 
    &=& 
    \m{\begin{qcircuit}[scale=0.5]
        \grid{3.5}{0,1}
        \controlled{\dotgate}{1,0}{1}
        \gate{$S$}{2.25,1}
      \end{qcircuit}
    }
    \\\nonumber\\[0ex]
    \m{\begin{qcircuit}[scale=0.5]
        \grid{3.5}{0,1}
        \gate{$S$}{1.25,0}
        \controlled{\dotgate}{2.5,0}{1}
      \end{qcircuit}
    } 
    &=& 
    \m{\begin{qcircuit}[scale=0.5]
        \grid{3.5}{0,1}
        \controlled{\dotgate}{1,0}{1}
        \gate{$S$}{2.25,0}
      \end{qcircuit}
    }
    \\\nonumber\\[0ex]
    \m{\begin{qcircuit}[scale=0.5]
        \grid{8}{0,1}
        \gate{$H$}{1.25,1}
        \gate{$S$}{2.75,1}
        \gate{$S$}{4.25,1}
        \gate{$H$}{5.75,1}
        \controlled{\dotgate}{7,0}{1}
      \end{qcircuit}
    } 
    &=& 
    \m{\begin{qcircuit}[scale=0.5]
        \grid{8}{0,1}
        \controlled{\dotgate}{1,0}{1}
        \gate{$S$}{2.25,0}
        \gate{$S$}{3.75,0}
        \gate{$H$}{2.25,1}
        \gate{$S$}{3.75,1}
        \gate{$S$}{5.25,1}
        \gate{$H$}{6.75,1}
      \end{qcircuit}
    }
    \\\nonumber\\[0ex]
    \m{\begin{qcircuit}[scale=0.5]
        \grid{8}{0,1}
        \gate{$H$}{1.25,0}
        \gate{$S$}{2.75,0}
        \gate{$S$}{4.25,0}
        \gate{$H$}{5.75,0}
        \controlled{\dotgate}{7,0}{1}
      \end{qcircuit}
    } 
    &=& 
    \m{\begin{qcircuit}[scale=0.5]
        \grid{8}{0,1}
        \controlled{\dotgate}{1,0}{1}
        \gate{$S$}{2.25,1}
        \gate{$S$}{3.75,1}
        \gate{$H$}{2.25,0}
        \gate{$S$}{3.75,0}
        \gate{$S$}{5.25,0}
        \gate{$H$}{6.75,0}
      \end{qcircuit}
    }
    \\\nonumber\\[0ex]
    \m{\begin{qcircuit}[scale=0.5]
        \grid{4.5}{0,1}
        \controlled{\dotgate}{1,0}{1}
        \gate{$H$}{2.25,1}
        \controlled{\dotgate}{3.5,0}{1}
      \end{qcircuit}
    }
    &=& 
    \m{\begin{qcircuit}[scale=0.5]
        \gridx{1.5}{10.5}{0,1}
        \gate{$S$}{2.75,1}
        \gate{$H$}{4.25,1}
        \controlled{\dotgate}{5.5,0}{1}
        \gate{$S$}{6.75,0}
        \gate{$S$}{6.75,1}
        \gate{$H$}{8.25,1}
        \gate{$S$}{9.75,1}
      \end{qcircuit}
    }\cdot\omega\inv
    \\\nonumber\\[0ex]
    \m{\begin{qcircuit}[scale=0.5]
        \grid{4.5}{0,1}
        \controlled{\dotgate}{1,0}{1}
        \gate{$H$}{2.25,0}
        \controlled{\dotgate}{3.5,0}{1}
      \end{qcircuit}
    }
    &=& 
    \m{\begin{qcircuit}[scale=0.5]
        \gridx{1.5}{10.5}{0,1}
        \gate{$S$}{2.75,0}
        \gate{$H$}{4.25,0}
        \controlled{\dotgate}{5.5,0}{1}
        \gate{$S$}{6.75,1}
        \gate{$S$}{6.75,0}
        \gate{$H$}{8.25,0}
        \gate{$S$}{9.75,0}
      \end{qcircuit}
    }\cdot\omega\inv
  \end{eqnarray}
  (d) Equations for $n\geq 3$:
  \begin{eqnarray}
    \m{\begin{qcircuit}[scale=0.5]
        \grid{3}{0,1,2}
        \controlled{\dotgate}{1,0}{1}
        \controlled{\dotgate}{2,1}{2}
      \end{qcircuit}
    } 
    &=& 
    \m{\begin{qcircuit}[scale=0.5]
        \grid{3}{0,1,2}
        \controlled{\dotgate}{1,1}{2}
        \controlled{\dotgate}{2,0}{1}
      \end{qcircuit}
    }
    \\\nonumber\\[0ex]
    \m{\begin{qcircuit}[scale=0.5]
        \grid{12.50}{0,1,2}
        \controlled{\dotgate}{1.25,2}{1}
        \gate{$H$}{2.75,2}
        \gate{$H$}{2.75,1}
        \controlled{\dotgate}{4.00,2}{1}
        \gate{$H$}{5.25,1}
        \gate{$H$}{5.25,0}
        \controlled{\dotgate}{6.50,1}{0}
        \gate{$H$}{7.75,1}
        \gate{$H$}{7.75,0}
        \controlled{\dotgate}{9.00,2}{1}
        \gate{$H$}{10.25,2}
        \gate{$H$}{10.25,1}
        \controlled{\dotgate}{11.50,2}{1}
      \end{qcircuit}
    }
    &=&
    \m{\begin{qcircuit}[scale=0.5]
        \grid{12.50}{0,1,2}
        \controlled{\dotgate}{1.25,0}{1}
        \gate{$H$}{2.75,0}
        \gate{$H$}{2.75,1}
        \controlled{\dotgate}{4.00,0}{1}
        \gate{$H$}{5.25,1}
        \gate{$H$}{5.25,2}
        \controlled{\dotgate}{6.50,1}{2}
        \gate{$H$}{7.75,1}
        \gate{$H$}{7.75,2}
        \controlled{\dotgate}{9.00,0}{1}
        \gate{$H$}{10.25,0}
        \gate{$H$}{10.25,1}
        \controlled{\dotgate}{11.50,0}{1}
      \end{qcircuit}
    }
  \end{eqnarray}
  \begin{eqnarray}
    \m{\begin{qcircuit}[scale=0.5]
        \grid{21.00}{0,1,2}
        \controlled{\dotgate}{1.00,1}{2}
        \gate{$H$}{2.25,2}
        \gate{$H$}{2.25,1}
        \controlled{\dotgate}{3.50,1}{2}
        \gate{$H$}{4.75,2}
        \gate{$H$}{4.75,1}
        \controlled{\dotgate}{6.00,0}{1}
        \controlled{\dotgate}{8.00,1}{2}
        \gate{$H$}{9.25,2}
        \gate{$H$}{9.25,1}
        \controlled{\dotgate}{10.50,1}{2}
        \gate{$H$}{11.75,2}
        \gate{$H$}{11.75,1}
        \controlled{\dotgate}{13.00,0}{1}
        \controlled{\dotgate}{15.00,1}{2}
        \gate{$H$}{16.25,2}
        \gate{$H$}{16.25,1}
        \controlled{\dotgate}{17.50,1}{2}
        \gate{$H$}{18.75,1}
        \gate{$H$}{18.75,2}
        \controlled{\dotgate}{20.00,0}{1}
      \end{qcircuit}
    }
    &=&
    \m{\begin{qcircuit}[scale=0.5]
        \grid{2.50}{0,1,2}
      \end{qcircuit}
    }
    \\\nonumber\\[0ex]
    \m{\begin{qcircuit}[scale=0.5]
        \grid{21.00}{0,1,2}
        \controlled{\dotgate}{1.00,1}{0}
        \gate{$H$}{2.25,0}
        \gate{$H$}{2.25,1}
        \controlled{\dotgate}{3.50,1}{0}
        \gate{$H$}{4.75,0}
        \gate{$H$}{4.75,1}
        \controlled{\dotgate}{6.00,2}{1}
        \controlled{\dotgate}{8.00,1}{0}
        \gate{$H$}{9.25,0}
        \gate{$H$}{9.25,1}
        \controlled{\dotgate}{10.50,1}{0}
        \gate{$H$}{11.75,0}
        \gate{$H$}{11.75,1}
        \controlled{\dotgate}{13.00,2}{1}
        \controlled{\dotgate}{15.00,1}{0}
        \gate{$H$}{16.25,0}
        \gate{$H$}{16.25,1}
        \controlled{\dotgate}{17.50,1}{0}
        \gate{$H$}{18.75,1}
        \gate{$H$}{18.75,0}
        \controlled{\dotgate}{20.00,2}{1}
      \end{qcircuit}
    }
    &=&
    \m{\begin{qcircuit}[scale=0.5]
        \grid{2.50}{0,1,2}
      \end{qcircuit}
    }
  \end{eqnarray}
  \caption{A presentation of the (strict spatial monoidal) Clifford
    groupoid by generators and relations}\label{fig-relations}

  \setcounter{equation}{\value{mytmpcounter}}
\end{figure}

\bibliographystyle{abbrv}
\bibliography{clifford}

\newpage\phantom{blah}\newpage\phantom{blah}
\end{document}

%% file: figures.tex
\def\defa{
\m{\begin{qcircuit}[scale=0.5]
    \grid{2.50}{0}
    \agate{1}{1.25}{0}
\end{qcircuit}
}
&=&
\m{\begin{qcircuit}[scale=0.5]
    \grid{1.00}{0}
\end{qcircuit}
}
\\\\[-1ex]
\m{\begin{qcircuit}[scale=0.5]
    \grid{2.50}{0}
    \agate{2}{1.25}{0}
\end{qcircuit}
}
&=&
\m{\begin{qcircuit}[scale=0.5]
    \grid{2.50}{0}
    \gate{$H$}{1.25,0}
\end{qcircuit}
}
\\\\[-1ex]
\m{\begin{qcircuit}[scale=0.5]
    \grid{2.50}{0}
    \agate{3}{1.25}{0}
\end{qcircuit}
}
&=&
\m{\begin{qcircuit}[scale=0.5]
    \grid{5.50}{0}
    \gate{$H$}{1.25,0}
    \gate{$S$}{2.75,0}
    \gate{$H$}{4.25,0}
\end{qcircuit}
}
\\\\[-1ex]
}
\def\defb{
\m{\begin{qcircuit}[scale=0.5]
    \grid{2.50}{0,1}
    \bgate{1}{1.25}{0}{1}
\end{qcircuit}
}
&=&
\m{\begin{qcircuit}[scale=0.5]
    \grid{8.50}{0,1}
    \gate{$H$}{1.25,0}
    \controlled{\dotgate}{2.50,1}{0}
    \gate{$H$}{3.75,0}
    \gate{$H$}{3.75,1}
    \controlled{\dotgate}{5.00,1}{0}
    \gate{$H$}{6.25,1}
    \gate{$H$}{6.25,0}
    \controlled{\dotgate}{7.50,1}{0}
\end{qcircuit}
}
\\\\[-1ex]
\m{\begin{qcircuit}[scale=0.5]
    \grid{2.50}{0,1}
    \bgate{2}{1.25}{0}{1}
\end{qcircuit}
}
&=&
\m{\begin{qcircuit}[scale=0.5]
    \grid{4.50}{0,1}
    \controlled{\dotgate}{1.00,1}{0}
    \gate{$H$}{2.25,1}
    \gate{$H$}{2.25,0}
    \controlled{\dotgate}{3.50,1}{0}
\end{qcircuit}
}
\\\\[-1ex]
\m{\begin{qcircuit}[scale=0.5]
    \grid{2.50}{0,1}
    \bgate{3}{1.25}{0}{1}
\end{qcircuit}
}
&=&
\m{\begin{qcircuit}[scale=0.5]
    \grid{7.50}{0,1}
    \gate{$H$}{1.25,1}
    \gate{$S$}{2.75,1}
    \controlled{\dotgate}{4.00,1}{0}
    \gate{$H$}{5.25,1}
    \gate{$H$}{5.25,0}
    \controlled{\dotgate}{6.50,1}{0}
\end{qcircuit}
}
\\\\[-1ex]
\m{\begin{qcircuit}[scale=0.5]
    \grid{2.50}{0,1}
    \bgate{4}{1.25}{0}{1}
\end{qcircuit}
}
&=&
\m{\begin{qcircuit}[scale=0.5]
    \grid{6.00}{0,1}
    \gate{$H$}{1.25,1}
    \controlled{\dotgate}{2.50,1}{0}
    \gate{$H$}{3.75,1}
    \gate{$H$}{3.75,0}
    \controlled{\dotgate}{5.00,1}{0}
\end{qcircuit}
}
\\\\[-1ex]
}
\def\defc{
\m{\begin{qcircuit}[scale=0.5]
    \grid{2.50}{0}
    \cgate{1}{1.25}{0}
\end{qcircuit}
}
&=&
\m{\begin{qcircuit}[scale=0.5]
    \grid{1.00}{0}
\end{qcircuit}
}
\\\\[-1ex]
\m{\begin{qcircuit}[scale=0.5]
    \grid{2.50}{0}
    \cgate{2}{1.25}{0}
\end{qcircuit}
}
&=&
\m{\begin{qcircuit}[scale=0.5]
    \grid{7.00}{0}
    \gate{$H$}{1.25,0}
    \gate{$S$}{2.75,0}
    \gate{$S$}{4.25,0}
    \gate{$H$}{5.75,0}
\end{qcircuit}
}
\\\\[-1ex]
}
\def\defd{
\m{\begin{qcircuit}[scale=0.5]
    \grid{2.50}{0,1}
    \dgate{1}{1.25}{0}{1}
\end{qcircuit}
}
&=&
\m{\begin{qcircuit}[scale=0.5]
    \grid{8.50}{0,1}
    \controlled{\dotgate}{1.00,1}{0}
    \gate{$H$}{2.25,1}
    \gate{$H$}{2.25,0}
    \controlled{\dotgate}{3.50,1}{0}
    \gate{$H$}{4.75,1}
    \gate{$H$}{4.75,0}
    \controlled{\dotgate}{6.00,1}{0}
    \gate{$H$}{7.25,0}
\end{qcircuit}
}
\\\\[-1ex]
\m{\begin{qcircuit}[scale=0.5]
    \grid{2.50}{0,1}
    \dgate{2}{1.25}{0}{1}
\end{qcircuit}
}
&=&
\m{\begin{qcircuit}[scale=0.5]
    \grid{7.50}{0,1}
    \gate{$H$}{1.25,1}
    \controlled{\dotgate}{2.50,1}{0}
    \gate{$H$}{3.75,1}
    \gate{$H$}{3.75,0}
    \controlled{\dotgate}{5.00,1}{0}
    \gate{$H$}{6.25,0}
\end{qcircuit}
}
\\\\[-1ex]
\m{\begin{qcircuit}[scale=0.5]
    \grid{2.50}{0,1}
    \dgate{3}{1.25}{0}{1}
\end{qcircuit}
}
&=&
\m{\begin{qcircuit}[scale=0.5]
    \grid{9.00}{0,1}
    \gate{$H$}{1.25,1}
    \gate{$H$}{1.25,0}
    \gate{$S$}{2.75,0}
    \controlled{\dotgate}{4.00,1}{0}
    \gate{$H$}{5.25,1}
    \gate{$H$}{5.25,0}
    \controlled{\dotgate}{6.50,1}{0}
    \gate{$H$}{7.75,0}
\end{qcircuit}
}
\\\\[-1ex]
\m{\begin{qcircuit}[scale=0.5]
    \grid{2.50}{0,1}
    \dgate{4}{1.25}{0}{1}
\end{qcircuit}
}
&=&
\m{\begin{qcircuit}[scale=0.5]
    \grid{7.50}{0,1}
    \gate{$H$}{1.25,1}
    \gate{$H$}{1.25,0}
    \controlled{\dotgate}{2.50,1}{0}
    \gate{$H$}{3.75,1}
    \gate{$H$}{3.75,0}
    \controlled{\dotgate}{5.00,1}{0}
    \gate{$H$}{6.25,0}
\end{qcircuit}
}
\\\\[-1ex]
}
\def\defe{
\m{\begin{qcircuit}[scale=0.5]
    \grid{2.50}{0}
    \egate{1}{1.25}{0}
\end{qcircuit}
}
&=&
\m{\begin{qcircuit}[scale=0.5]
    \grid{1.00}{0}
\end{qcircuit}
}
\\\\[-1ex]
\m{\begin{qcircuit}[scale=0.5]
    \grid{2.50}{0}
    \egate{2}{1.25}{0}
\end{qcircuit}
}
&=&
\m{\begin{qcircuit}[scale=0.5]
    \grid{2.50}{0}
    \gate{$S$}{1.25,0}
\end{qcircuit}
}
\\\\[-1ex]
\m{\begin{qcircuit}[scale=0.5]
    \grid{2.50}{0}
    \egate{3}{1.25}{0}
\end{qcircuit}
}
&=&
\m{\begin{qcircuit}[scale=0.5]
    \grid{4.00}{0}
    \gate{$S$}{1.25,0}
    \gate{$S$}{2.75,0}
\end{qcircuit}
}
\\\\[-1ex]
\m{\begin{qcircuit}[scale=0.5]
    \grid{2.50}{0}
    \egate{4}{1.25}{0}
\end{qcircuit}
}
&=&
\m{\begin{qcircuit}[scale=0.5]
    \grid{5.50}{0}
    \gate{$S$}{1.25,0}
    \gate{$S$}{2.75,0}
    \gate{$S$}{4.25,0}
\end{qcircuit}
}
\\\\[-1ex]
}

\def\paulia{\begin{qcircuit}[scale=0.5]
    \grid{2.50}{0}
    \leftlabel{$Z$}{0.00,0}
    \agate{1}{1.25}{0}
    \rightlabel{$Z$}{2.50,0}
\end{qcircuit}
\\\\[-1ex]
\begin{qcircuit}[scale=0.5]
    \grid{2.50}{0}
    \leftlabel{$X$}{0.00,0}
    \agate{2}{1.25}{0}
    \rightlabel{$Z$}{2.50,0}
\end{qcircuit}
\\\\[-1ex]
\begin{qcircuit}[scale=0.5]
    \grid{2.50}{0}
    \leftlabel{$Y$}{0.00,0}
    \agate{3}{1.25}{0}
    \rightlabel{$Z$}{2.50,0}
\end{qcircuit}
\\\\[-1ex]
}
\def\paulib{\begin{qcircuit}[scale=0.5]
    \grid{2.50}{0,1}
    \leftlabel{$I$}{0.00,1}
    \leftlabel{$Z$}{0.00,0}
    \bgate{1}{1.25}{0}{1}
    \rightlabel{$Z$}{2.50,1}
    \rightlabel{$I$}{2.50,0}
\end{qcircuit}
\\\\[-1ex]
\begin{qcircuit}[scale=0.5]
    \grid{2.50}{0,1}
    \leftlabel{$X$}{0.00,1}
    \leftlabel{$Z$}{0.00,0}
    \bgate{2}{1.25}{0}{1}
    \rightlabel{$Z$}{2.50,1}
    \rightlabel{$I$}{2.50,0}
\end{qcircuit}
\\\\[-1ex]
\begin{qcircuit}[scale=0.5]
    \grid{2.50}{0,1}
    \leftlabel{$Y$}{0.00,1}
    \leftlabel{$Z$}{0.00,0}
    \bgate{3}{1.25}{0}{1}
    \rightlabel{$Z$}{2.50,1}
    \rightlabel{$I$}{2.50,0}
\end{qcircuit}
\\\\[-1ex]
\begin{qcircuit}[scale=0.5]
    \grid{2.50}{0,1}
    \leftlabel{$Z$}{0.00,1}
    \leftlabel{$Z$}{0.00,0}
    \bgate{4}{1.25}{0}{1}
    \rightlabel{$Z$}{2.50,1}
    \rightlabel{$I$}{2.50,0}
\end{qcircuit}
\\\\[-1ex]
}
\def\paulic{\begin{qcircuit}[scale=0.5]
    \grid{2.50}{0}
    \leftlabel{$Z$}{0.00,0}
    \cgate{1}{1.25}{0}
    \rightlabel{$Z$}{2.50,0}
\end{qcircuit}
\\\\[-1ex]
\begin{qcircuit}[scale=0.5]
    \grid{2.50}{0}
    \leftlabel{$-Z$}{0.00,0}
    \cgate{2}{1.25}{0}
    \rightlabel{$Z$}{2.50,0}
\end{qcircuit}
\\\\[-1ex]
}

\def\pauliez{\begin{qcircuit}[scale=0.5]
    \grid{2.50}{0}
    \leftlabel{$Z$}{0.00,0}
    \egate{1}{1.25}{0}
    \rightlabel{$Z$}{2.50,0}
\end{qcircuit}
\\\\[-1ex]
\begin{qcircuit}[scale=0.5]
    \grid{2.50}{0}
    \leftlabel{$Z$}{0.00,0}
    \egate{2}{1.25}{0}
    \rightlabel{$Z$}{2.50,0}
\end{qcircuit}
\\\\[-1ex]
\begin{qcircuit}[scale=0.5]
    \grid{2.50}{0}
    \leftlabel{$Z$}{0.00,0}
    \egate{3}{1.25}{0}
    \rightlabel{$Z$}{2.50,0}
\end{qcircuit}
\\\\[-1ex]
\begin{qcircuit}[scale=0.5]
    \grid{2.50}{0}
    \leftlabel{$Z$}{0.00,0}
    \egate{4}{1.25}{0}
    \rightlabel{$Z$}{2.50,0}
\end{qcircuit}
\\\\[-1ex]
}
\def\pauliex{\begin{qcircuit}[scale=0.5]
    \grid{2.50}{0}
    \leftlabel{$X$}{0.00,0}
    \egate{1}{1.25}{0}
    \rightlabel{$X$}{2.50,0}
\end{qcircuit}
\\\\[-1ex]
\begin{qcircuit}[scale=0.5]
    \grid{2.50}{0}
    \leftlabel{$-Y$}{0.00,0}
    \egate{2}{1.25}{0}
    \rightlabel{$X$}{2.50,0}
\end{qcircuit}
\\\\[-1ex]
\begin{qcircuit}[scale=0.5]
    \grid{2.50}{0}
    \leftlabel{$-X$}{0.00,0}
    \egate{3}{1.25}{0}
    \rightlabel{$X$}{2.50,0}
\end{qcircuit}
\\\\[-1ex]
\begin{qcircuit}[scale=0.5]
    \grid{2.50}{0}
    \leftlabel{$Y$}{0.00,0}
    \egate{4}{1.25}{0}
    \rightlabel{$X$}{2.50,0}
\end{qcircuit}
\\\\[-1ex]
}
\def\paulifz{\begin{qcircuit}[scale=0.5]
    \grid{2.50}{0,1}
    \leftlabel{$Z$}{0.00,1}
    \leftlabel{$I$}{0.00,0}
    \dgate{1}{1.25}{0}{1}
    \rightlabel{$I$}{2.50,1}
    \rightlabel{$Z$}{2.50,0}
\end{qcircuit}
\\\\[-1ex]
\begin{qcircuit}[scale=0.5]
    \grid{2.50}{0,1}
    \leftlabel{$Z$}{0.00,1}
    \leftlabel{$I$}{0.00,0}
    \dgate{2}{1.25}{0}{1}
    \rightlabel{$I$}{2.50,1}
    \rightlabel{$Z$}{2.50,0}
\end{qcircuit}
\\\\[-1ex]
\begin{qcircuit}[scale=0.5]
    \grid{2.50}{0,1}
    \leftlabel{$Z$}{0.00,1}
    \leftlabel{$I$}{0.00,0}
    \dgate{3}{1.25}{0}{1}
    \rightlabel{$I$}{2.50,1}
    \rightlabel{$Z$}{2.50,0}
\end{qcircuit}
\\\\[-1ex]
\begin{qcircuit}[scale=0.5]
    \grid{2.50}{0,1}
    \leftlabel{$Z$}{0.00,1}
    \leftlabel{$I$}{0.00,0}
    \dgate{4}{1.25}{0}{1}
    \rightlabel{$I$}{2.50,1}
    \rightlabel{$Z$}{2.50,0}
\end{qcircuit}
\\\\[-1ex]
}
\def\paulifx{\begin{qcircuit}[scale=0.5]
    \grid{2.50}{0,1}
    \leftlabel{$X$}{0.00,1}
    \leftlabel{$I$}{0.00,0}
    \dgate{1}{1.25}{0}{1}
    \rightlabel{$I$}{2.50,1}
    \rightlabel{$X$}{2.50,0}
\end{qcircuit}
\\\\[-1ex]
\begin{qcircuit}[scale=0.5]
    \grid{2.50}{0,1}
    \leftlabel{$X$}{0.00,1}
    \leftlabel{$X$}{0.00,0}
    \dgate{2}{1.25}{0}{1}
    \rightlabel{$I$}{2.50,1}
    \rightlabel{$X$}{2.50,0}
\end{qcircuit}
\\\\[-1ex]
\begin{qcircuit}[scale=0.5]
    \grid{2.50}{0,1}
    \leftlabel{$X$}{0.00,1}
    \leftlabel{$Y$}{0.00,0}
    \dgate{3}{1.25}{0}{1}
    \rightlabel{$I$}{2.50,1}
    \rightlabel{$X$}{2.50,0}
\end{qcircuit}
\\\\[-1ex]
\begin{qcircuit}[scale=0.5]
    \grid{2.50}{0,1}
    \leftlabel{$X$}{0.00,1}
    \leftlabel{$Z$}{0.00,0}
    \dgate{4}{1.25}{0}{1}
    \rightlabel{$I$}{2.50,1}
    \rightlabel{$X$}{2.50,0}
\end{qcircuit}
\\\\[-1ex]
}
\def\paulify{\begin{qcircuit}[scale=0.5]
    \grid{2.50}{0,1}
    \leftlabel{$Y$}{0.00,1}
    \leftlabel{$I$}{0.00,0}
    \dgate{1}{1.25}{0}{1}
    \rightlabel{$I$}{2.50,1}
    \rightlabel{$Y$}{2.50,0}
\end{qcircuit}
\\\\[-1ex]
\begin{qcircuit}[scale=0.5]
    \grid{2.50}{0,1}
    \leftlabel{$Y$}{0.00,1}
    \leftlabel{$X$}{0.00,0}
    \dgate{2}{1.25}{0}{1}
    \rightlabel{$I$}{2.50,1}
    \rightlabel{$Y$}{2.50,0}
\end{qcircuit}
\\\\[-1ex]
\begin{qcircuit}[scale=0.5]
    \grid{2.50}{0,1}
    \leftlabel{$Y$}{0.00,1}
    \leftlabel{$Y$}{0.00,0}
    \dgate{3}{1.25}{0}{1}
    \rightlabel{$I$}{2.50,1}
    \rightlabel{$Y$}{2.50,0}
\end{qcircuit}
\\\\[-1ex]
\begin{qcircuit}[scale=0.5]
    \grid{2.50}{0,1}
    \leftlabel{$Y$}{0.00,1}
    \leftlabel{$Z$}{0.00,0}
    \dgate{4}{1.25}{0}{1}
    \rightlabel{$I$}{2.50,1}
    \rightlabel{$Y$}{2.50,0}
\end{qcircuit}
\\\\[-1ex]
}
\def\commHA{
\m{\begin{qcircuit}[scale=0.5]
    \grid{4.00}{0}
    \gate{$H$}{1.25,0}
    \agate{1}{2.75}{0}
\end{qcircuit}
}
&=&
\m{\begin{qcircuit}[scale=0.5]
    \grid{2.50}{0}
    \agate{2}{1.25}{0}
\end{qcircuit}
}
\\\\[-1ex]
\m{\begin{qcircuit}[scale=0.5]
    \grid{4.00}{0}
    \gate{$H$}{1.25,0}
    \agate{2}{2.75}{0}
\end{qcircuit}
}
&=&
\m{\begin{qcircuit}[scale=0.5]
    \grid{2.50}{0}
    \agate{1}{1.25}{0}
\end{qcircuit}
}
\\\\[-1ex]
\m{\begin{qcircuit}[scale=0.5]
    \grid{4.00}{0}
    \gate{$H$}{1.25,0}
    \agate{3}{2.75}{0}
\end{qcircuit}
}
&=&
\m{\begin{qcircuit}[scale=0.5]
    \grid{8.50}{0}
    \agate{3}{1.25}{0}
    \gate{$X$}{2.75,0}
    \gate{$S$}{4.25,0}
    \gate{$S$}{5.75,0}
    \gate{$S$}{7.25,0}
\end{qcircuit}
}
\cdot \omega
\\\\[-1ex]
}
\def\commSA{
\m{\begin{qcircuit}[scale=0.5]
    \grid{4.00}{0}
    \gate{$S$}{1.25,0}
    \agate{1}{2.75}{0}
\end{qcircuit}
}
&=&
\m{\begin{qcircuit}[scale=0.5]
    \grid{4.00}{0}
    \agate{1}{1.25}{0}
    \gate{$S$}{2.75,0}
\end{qcircuit}
}
\\\\[-1ex]
\m{\begin{qcircuit}[scale=0.5]
    \grid{4.00}{0}
    \gate{$S$}{1.25,0}
    \agate{2}{2.75}{0}
\end{qcircuit}
}
&=&
\m{\begin{qcircuit}[scale=0.5]
    \grid{8.50}{0}
    \agate{3}{1.25}{0}
    \gate{$X$}{2.75,0}
    \gate{$S$}{4.25,0}
    \gate{$S$}{5.75,0}
    \gate{$S$}{7.25,0}
\end{qcircuit}
}
\cdot \omega
\\\\[-1ex]
\m{\begin{qcircuit}[scale=0.5]
    \grid{4.00}{0}
    \gate{$S$}{1.25,0}
    \agate{3}{2.75}{0}
\end{qcircuit}
}
&=&
\m{\begin{qcircuit}[scale=0.5]
    \grid{7.00}{0}
    \agate{2}{1.25}{0}
    \gate{$S$}{2.75,0}
    \gate{$S$}{4.25,0}
    \gate{$S$}{5.75,0}
\end{qcircuit}
}
\cdot \omega
\\\\[-1ex]
}
\def\commZZAI{
\m{\begin{qcircuit}[scale=0.5]
    \grid{3.50}{0,1}
    \controlled{\dotgate}{1.00,1}{0}
    \agate{1}{2.25}{1}
\end{qcircuit}
}
&=&
\m{\begin{qcircuit}[scale=0.5]
    \grid{3.50}{0,1}
    \agate{1}{1.25}{1}
    \controlled{\dotgate}{2.50,1}{0}
\end{qcircuit}
}
\\\\[-1ex]
\m{\begin{qcircuit}[scale=0.5]
    \grid{3.50}{0,1}
    \controlled{\dotgate}{1.00,1}{0}
    \agate{2}{2.25}{1}
\end{qcircuit}
}
&=&
\m{\begin{qcircuit}[scale=0.5]
    \grid{6.50}{0,1}
    \agate{1}{1.25}{0}
    \bgate{2}{2.75}{0}{1}
    \controlled{\dotgate}{4.00,1}{0}
    \gate{$H$}{5.25,0}
\end{qcircuit}
}
\\\\[-1ex]
\m{\begin{qcircuit}[scale=0.5]
    \grid{3.50}{0,1}
    \controlled{\dotgate}{1.00,1}{0}
    \agate{3}{2.25}{1}
\end{qcircuit}
}
&=&
\m{\begin{qcircuit}[scale=0.5]
    \grid{14.00}{0,1}
    \agate{1}{1.25}{0}
    \bgate{3}{2.75}{0}{1}
    \gate{$H$}{4.25,0}
    \gate{$S$}{5.75,0}
    \gate{$H$}{7.25,0}
    \controlled{\dotgate}{8.50,1}{0}
    \gate{$S$}{9.75,1}
    \gate{$S$}{11.25,1}
    \gate{$S$}{12.75,1}
    \gate{$H$}{9.75,0}
\end{qcircuit}
}
\\\\[-1ex]
}
\def\commZZIABB{
\m{\begin{qcircuit}[scale=0.5]
    \grid{5.00}{0,1}
    \controlled{\dotgate}{1.00,1}{0}
    \agate{1}{2.25}{0}
    \bgate{1}{3.75}{0}{1}
\end{qcircuit}
}
&=&
\m{\begin{qcircuit}[scale=0.5]
    \grid{8.00}{0,1}
    \agate{1}{1.25}{0}
    \bgate{1}{2.75}{0}{1}
    \gate{$H$}{4.25,0}
    \controlled{\dotgate}{5.50,1}{0}
    \gate{$H$}{6.75,0}
\end{qcircuit}
}
\\\\[-1ex]
\m{\begin{qcircuit}[scale=0.5]
    \grid{5.00}{0,1}
    \controlled{\dotgate}{1.00,1}{0}
    \agate{1}{2.25}{0}
    \bgate{2}{3.75}{0}{1}
\end{qcircuit}
}
&=&
\m{\begin{qcircuit}[scale=0.5]
    \grid{5.00}{0,1}
    \agate{2}{1.25}{1}
    \gate{$H$}{2.75,0}
    \controlled{\dotgate}{4.00,1}{0}
\end{qcircuit}
}
\\\\[-1ex]
\m{\begin{qcircuit}[scale=0.5]
    \grid{5.00}{0,1}
    \controlled{\dotgate}{1.00,1}{0}
    \agate{1}{2.25}{0}
    \bgate{3}{3.75}{0}{1}
\end{qcircuit}
}
&=&
\m{\begin{qcircuit}[scale=0.5]
    \grid{9.50}{0,1}
    \agate{3}{1.25}{1}
    \gate{$H$}{2.75,0}
    \controlled{\dotgate}{4.00,1}{0}
    \gate{$S$}{5.25,1}
    \gate{$S$}{5.25,0}
    \gate{$H$}{6.75,0}
    \gate{$S$}{8.25,0}
\end{qcircuit}
}
\cdot \omega^{7}
\\\\[-1ex]
\m{\begin{qcircuit}[scale=0.5]
    \grid{5.00}{0,1}
    \controlled{\dotgate}{1.00,1}{0}
    \agate{1}{2.25}{0}
    \bgate{4}{3.75}{0}{1}
\end{qcircuit}
}
&=&
\m{\begin{qcircuit}[scale=0.5]
    \grid{11.00}{0,1}
    \agate{1}{1.25}{0}
    \bgate{4}{2.75}{0}{1}
    \gate{$H$}{4.25,0}
    \controlled{\dotgate}{5.50,1}{0}
    \gate{$S$}{6.75,0}
    \gate{$S$}{8.25,0}
    \gate{$H$}{9.75,0}
\end{qcircuit}
}
\\\\[-1ex]
\m{\begin{qcircuit}[scale=0.5]
    \grid{5.00}{0,1}
    \controlled{\dotgate}{1.00,1}{0}
    \agate{2}{2.25}{0}
    \bgate{1}{3.75}{0}{1}
\end{qcircuit}
}
&=&
\m{\begin{qcircuit}[scale=0.5]
    \grid{4.00}{0,1}
    \agate{2}{1.25}{0}
    \bgate{4}{2.75}{0}{1}
\end{qcircuit}
}
\\\\[-1ex]
\m{\begin{qcircuit}[scale=0.5]
    \grid{5.00}{0,1}
    \controlled{\dotgate}{1.00,1}{0}
    \agate{2}{2.25}{0}
    \bgate{2}{3.75}{0}{1}
\end{qcircuit}
}
&=&
\m{\begin{qcircuit}[scale=0.5]
    \grid{12.50}{0,1}
    \agate{3}{1.25}{0}
    \bgate{3}{2.75}{0}{1}
    \gate{$H$}{4.25,0}
    \controlled{\dotgate}{5.50,1}{0}
    \gate{$S$}{6.75,1}
    \gate{$S$}{6.75,0}
    \gate{$S$}{8.25,0}
    \gate{$H$}{9.75,0}
    \gate{$S$}{11.25,0}
\end{qcircuit}
}
\cdot \omega^{6}
\\\\[-1ex]
\m{\begin{qcircuit}[scale=0.5]
    \grid{5.00}{0,1}
    \controlled{\dotgate}{1.00,1}{0}
    \agate{2}{2.25}{0}
    \bgate{3}{3.75}{0}{1}
\end{qcircuit}
}
&=&
\m{\begin{qcircuit}[scale=0.5]
    \grid{11.00}{0,1}
    \agate{3}{1.25}{0}
    \bgate{2}{2.75}{0}{1}
    \gate{$X$}{4.25,1}
    \gate{$H$}{4.25,0}
    \controlled{\dotgate}{5.50,1}{0}
    \gate{$S$}{6.75,1}
    \gate{$S$}{8.25,1}
    \gate{$S$}{9.75,1}
\end{qcircuit}
}
\cdot \omega
\\\\[-1ex]
\m{\begin{qcircuit}[scale=0.5]
    \grid{5.00}{0,1}
    \controlled{\dotgate}{1.00,1}{0}
    \agate{2}{2.25}{0}
    \bgate{4}{3.75}{0}{1}
\end{qcircuit}
}
&=&
\m{\begin{qcircuit}[scale=0.5]
    \grid{4.00}{0,1}
    \agate{2}{1.25}{0}
    \bgate{1}{2.75}{0}{1}
\end{qcircuit}
}
\\\\[-1ex]
\m{\begin{qcircuit}[scale=0.5]
    \grid{5.00}{0,1}
    \controlled{\dotgate}{1.00,1}{0}
    \agate{3}{2.25}{0}
    \bgate{1}{3.75}{0}{1}
\end{qcircuit}
}
&=&
\m{\begin{qcircuit}[scale=0.5]
    \grid{9.50}{0,1}
    \agate{3}{1.25}{0}
    \bgate{4}{2.75}{0}{1}
    \gate{$H$}{4.25,0}
    \controlled{\dotgate}{5.50,1}{0}
    \gate{$S$}{6.75,0}
    \gate{$H$}{8.25,0}
\end{qcircuit}
}
\\\\[-1ex]
\m{\begin{qcircuit}[scale=0.5]
    \grid{5.00}{0,1}
    \controlled{\dotgate}{1.00,1}{0}
    \agate{3}{2.25}{0}
    \bgate{2}{3.75}{0}{1}
\end{qcircuit}
}
&=&
\m{\begin{qcircuit}[scale=0.5]
    \grid{11.00}{0,1}
    \agate{2}{1.25}{0}
    \bgate{3}{2.75}{0}{1}
    \gate{$X$}{4.25,1}
    \controlled{\dotgate}{5.50,1}{0}
    \gate{$S$}{6.75,1}
    \gate{$S$}{8.25,1}
    \gate{$S$}{9.75,1}
    \gate{$S$}{6.75,0}
    \gate{$S$}{8.25,0}
    \gate{$H$}{9.75,0}
\end{qcircuit}
}
\cdot \omega
\\\\[-1ex]
\m{\begin{qcircuit}[scale=0.5]
    \grid{5.00}{0,1}
    \controlled{\dotgate}{1.00,1}{0}
    \agate{3}{2.25}{0}
    \bgate{3}{3.75}{0}{1}
\end{qcircuit}
}
&=&
\m{\begin{qcircuit}[scale=0.5]
    \grid{15.50}{0,1}
    \agate{2}{1.25}{0}
    \bgate{2}{2.75}{0}{1}
    \gate{$H$}{4.25,0}
    \gate{$S$}{5.75,0}
    \gate{$H$}{7.25,0}
    \controlled{\dotgate}{8.50,1}{0}
    \gate{$S$}{9.75,1}
    \gate{$S$}{11.25,1}
    \gate{$S$}{12.75,1}
    \gate{$S$}{9.75,0}
    \gate{$S$}{11.25,0}
    \gate{$S$}{12.75,0}
    \gate{$H$}{14.25,0}
\end{qcircuit}
}
\cdot \omega
\\\\[-1ex]
\m{\begin{qcircuit}[scale=0.5]
    \grid{5.00}{0,1}
    \controlled{\dotgate}{1.00,1}{0}
    \agate{3}{2.25}{0}
    \bgate{4}{3.75}{0}{1}
\end{qcircuit}
}
&=&
\m{\begin{qcircuit}[scale=0.5]
    \grid{12.50}{0,1}
    \agate{3}{1.25}{0}
    \bgate{1}{2.75}{0}{1}
    \gate{$H$}{4.25,0}
    \controlled{\dotgate}{5.50,1}{0}
    \gate{$S$}{6.75,0}
    \gate{$S$}{8.25,0}
    \gate{$S$}{9.75,0}
    \gate{$H$}{11.25,0}
\end{qcircuit}
}
\\\\[-1ex]
}
\def\commHIB{
\m{\begin{qcircuit}[scale=0.5]
    \grid{4.00}{0,1}
    \gate{$H$}{1.25,1}
    \bgate{1}{2.75}{0}{1}
\end{qcircuit}
}
&=&
\m{\begin{qcircuit}[scale=0.5]
    \grid{4.00}{0,1}
    \bgate{1}{1.25}{0}{1}
    \gate{$H$}{2.75,0}
\end{qcircuit}
}
\\\\[-1ex]
\m{\begin{qcircuit}[scale=0.5]
    \grid{4.00}{0,1}
    \gate{$H$}{1.25,1}
    \bgate{2}{2.75}{0}{1}
\end{qcircuit}
}
&=&
\m{\begin{qcircuit}[scale=0.5]
    \grid{2.50}{0,1}
    \bgate{4}{1.25}{0}{1}
\end{qcircuit}
}
\\\\[-1ex]
\m{\begin{qcircuit}[scale=0.5]
    \grid{4.00}{0,1}
    \gate{$H$}{1.25,1}
    \bgate{3}{2.75}{0}{1}
\end{qcircuit}
}
&=&
\m{\begin{qcircuit}[scale=0.5]
    \grid{10.00}{0,1}
    \bgate{3}{1.25}{0}{1}
    \gate{$X$}{2.75,1}
    \gate{$S$}{2.75,0}
    \gate{$S$}{4.25,0}
    \gate{$S$}{5.75,0}
    \gate{$H$}{7.25,0}
    \gate{$S$}{8.75,0}
\end{qcircuit}
}
\\\\[-1ex]
\m{\begin{qcircuit}[scale=0.5]
    \grid{4.00}{0,1}
    \gate{$H$}{1.25,1}
    \bgate{4}{2.75}{0}{1}
\end{qcircuit}
}
&=&
\m{\begin{qcircuit}[scale=0.5]
    \grid{2.50}{0,1}
    \bgate{2}{1.25}{0}{1}
\end{qcircuit}
}
\\\\[-1ex]
}
\def\commSIB{
\m{\begin{qcircuit}[scale=0.5]
    \grid{4.00}{0,1}
    \gate{$S$}{1.25,1}
    \bgate{1}{2.75}{0}{1}
\end{qcircuit}
}
&=&
\m{\begin{qcircuit}[scale=0.5]
    \grid{7.00}{0,1}
    \bgate{1}{1.25}{0}{1}
    \gate{$H$}{2.75,0}
    \gate{$S$}{4.25,0}
    \gate{$H$}{5.75,0}
\end{qcircuit}
}
\\\\[-1ex]
\m{\begin{qcircuit}[scale=0.5]
    \grid{4.00}{0,1}
    \gate{$S$}{1.25,1}
    \bgate{2}{2.75}{0}{1}
\end{qcircuit}
}
&=&
\m{\begin{qcircuit}[scale=0.5]
    \grid{10.00}{0,1}
    \bgate{3}{1.25}{0}{1}
    \gate{$X$}{2.75,1}
    \gate{$S$}{2.75,0}
    \gate{$S$}{4.25,0}
    \gate{$S$}{5.75,0}
    \gate{$H$}{7.25,0}
    \gate{$S$}{8.75,0}
\end{qcircuit}
}
\\\\[-1ex]
\m{\begin{qcircuit}[scale=0.5]
    \grid{4.00}{0,1}
    \gate{$S$}{1.25,1}
    \bgate{3}{2.75}{0}{1}
\end{qcircuit}
}
&=&
\m{\begin{qcircuit}[scale=0.5]
    \grid{7.00}{0,1}
    \bgate{2}{1.25}{0}{1}
    \gate{$S$}{2.75,0}
    \gate{$H$}{4.25,0}
    \gate{$S$}{5.75,0}
\end{qcircuit}
}
\\\\[-1ex]
\m{\begin{qcircuit}[scale=0.5]
    \grid{4.00}{0,1}
    \gate{$S$}{1.25,1}
    \bgate{4}{2.75}{0}{1}
\end{qcircuit}
}
&=&
\m{\begin{qcircuit}[scale=0.5]
    \grid{7.00}{0,1}
    \bgate{4}{1.25}{0}{1}
    \gate{$H$}{2.75,0}
    \gate{$S$}{4.25,0}
    \gate{$H$}{5.75,0}
\end{qcircuit}
}
\\\\[-1ex]
}
\def\commISB{
\m{\begin{qcircuit}[scale=0.5]
    \grid{4.00}{0,1}
    \gate{$S$}{1.25,0}
    \bgate{1}{2.75}{0}{1}
\end{qcircuit}
}
&=&
\m{\begin{qcircuit}[scale=0.5]
    \grid{4.00}{0,1}
    \bgate{1}{1.25}{0}{1}
    \gate{$S$}{2.75,1}
\end{qcircuit}
}
\\\\[-1ex]
\m{\begin{qcircuit}[scale=0.5]
    \grid{4.00}{0,1}
    \gate{$S$}{1.25,0}
    \bgate{2}{2.75}{0}{1}
\end{qcircuit}
}
&=&
\m{\begin{qcircuit}[scale=0.5]
    \grid{8.00}{0,1}
    \bgate{2}{1.25}{0}{1}
    \gate{$H$}{2.75,0}
    \controlled{\dotgate}{4.00,1}{0}
    \gate{$S$}{5.25,1}
    \gate{$S$}{5.25,0}
    \gate{$H$}{6.75,0}
\end{qcircuit}
}
\\\\[-1ex]
\m{\begin{qcircuit}[scale=0.5]
    \grid{4.00}{0,1}
    \gate{$S$}{1.25,0}
    \bgate{3}{2.75}{0}{1}
\end{qcircuit}
}
&=&
\m{\begin{qcircuit}[scale=0.5]
    \grid{8.00}{0,1}
    \bgate{3}{1.25}{0}{1}
    \gate{$H$}{2.75,0}
    \controlled{\dotgate}{4.00,1}{0}
    \gate{$S$}{5.25,1}
    \gate{$S$}{5.25,0}
    \gate{$H$}{6.75,0}
\end{qcircuit}
}
\\\\[-1ex]
\m{\begin{qcircuit}[scale=0.5]
    \grid{4.00}{0,1}
    \gate{$S$}{1.25,0}
    \bgate{4}{2.75}{0}{1}
\end{qcircuit}
}
&=&
\m{\begin{qcircuit}[scale=0.5]
    \grid{8.00}{0,1}
    \bgate{4}{1.25}{0}{1}
    \gate{$H$}{2.75,0}
    \controlled{\dotgate}{4.00,1}{0}
    \gate{$S$}{5.25,1}
    \gate{$S$}{5.25,0}
    \gate{$H$}{6.75,0}
\end{qcircuit}
}
\\\\[-1ex]
}
\def\commIXB{
\m{\begin{qcircuit}[scale=0.5]
    \grid{4.00}{0,1}
    \gate{$X$}{1.25,0}
    \bgate{1}{2.75}{0}{1}
\end{qcircuit}
}
&=&
\m{\begin{qcircuit}[scale=0.5]
    \grid{4.00}{0,1}
    \bgate{1}{1.25}{0}{1}
    \gate{$X$}{2.75,1}
\end{qcircuit}
}
\\\\[-1ex]
\m{\begin{qcircuit}[scale=0.5]
    \grid{4.00}{0,1}
    \gate{$X$}{1.25,0}
    \bgate{2}{2.75}{0}{1}
\end{qcircuit}
}
&=&
\m{\begin{qcircuit}[scale=0.5]
    \grid{4.00}{0,1}
    \bgate{2}{1.25}{0}{1}
    \gate{$X$}{2.75,1}
\end{qcircuit}
}
\\\\[-1ex]
\m{\begin{qcircuit}[scale=0.5]
    \grid{4.00}{0,1}
    \gate{$X$}{1.25,0}
    \bgate{3}{2.75}{0}{1}
\end{qcircuit}
}
&=&
\m{\begin{qcircuit}[scale=0.5]
    \grid{4.00}{0,1}
    \bgate{3}{1.25}{0}{1}
    \gate{$X$}{2.75,1}
\end{qcircuit}
}
\\\\[-1ex]
\m{\begin{qcircuit}[scale=0.5]
    \grid{4.00}{0,1}
    \gate{$X$}{1.25,0}
    \bgate{4}{2.75}{0}{1}
\end{qcircuit}
}
&=&
\m{\begin{qcircuit}[scale=0.5]
    \grid{4.00}{0,1}
    \bgate{4}{1.25}{0}{1}
    \gate{$X$}{2.75,1}
\end{qcircuit}
}
\\\\[-1ex]
}
\def\commIZZBBI{
\m{\begin{qcircuit}[scale=0.5]
    \grid{3.50}{0,1,2}
    \controlled{\dotgate}{1.00,1}{0}
    \bgate{1}{2.25}{1}{2}
\end{qcircuit}
}
&=&
\m{\begin{qcircuit}[scale=0.5]
    \grid{12.50}{0,1,2}
    \bgate{1}{1.25}{1}{2}
    \controlled{\dotgate}{2.50,2}{1}
    \gate{$H$}{3.75,1}
    \controlled{\dotgate}{5.00,1}{0}
    \gate{$H$}{6.25,1}
    \controlled{\dotgate}{7.50,2}{1}
    \gate{$H$}{8.75,1}
    \controlled{\dotgate}{10.00,1}{0}
    \gate{$H$}{11.25,1}
\end{qcircuit}
}
\\\\[-1ex]
\m{\begin{qcircuit}[scale=0.5]
    \grid{3.50}{0,1,2}
    \controlled{\dotgate}{1.00,1}{0}
    \bgate{2}{2.25}{1}{2}
\end{qcircuit}
}
&=&
\m{\begin{qcircuit}[scale=0.5]
    \grid{8.50}{0,1,2}
    \bgate{2}{1.25}{1}{2}
    \controlled{\dotgate}{2.50,2}{1}
    \gate{$H$}{3.75,1}
    \controlled{\dotgate}{5.00,1}{0}
    \gate{$H$}{6.25,1}
    \controlled{\dotgate}{7.50,2}{1}
\end{qcircuit}
}
\\\\[-1ex]
\m{\begin{qcircuit}[scale=0.5]
    \grid{3.50}{0,1,2}
    \controlled{\dotgate}{1.00,1}{0}
    \bgate{3}{2.25}{1}{2}
\end{qcircuit}
}
&=&
\m{\begin{qcircuit}[scale=0.5]
    \grid{8.50}{0,1,2}
    \bgate{3}{1.25}{1}{2}
    \controlled{\dotgate}{2.50,2}{1}
    \gate{$H$}{3.75,1}
    \controlled{\dotgate}{5.00,1}{0}
    \gate{$H$}{6.25,1}
    \controlled{\dotgate}{7.50,2}{1}
\end{qcircuit}
}
\\\\[-1ex]
\m{\begin{qcircuit}[scale=0.5]
    \grid{3.50}{0,1,2}
    \controlled{\dotgate}{1.00,1}{0}
    \bgate{4}{2.25}{1}{2}
\end{qcircuit}
}
&=&
\m{\begin{qcircuit}[scale=0.5]
    \grid{8.50}{0,1,2}
    \bgate{4}{1.25}{1}{2}
    \controlled{\dotgate}{2.50,2}{1}
    \gate{$H$}{3.75,1}
    \controlled{\dotgate}{5.00,1}{0}
    \gate{$H$}{6.25,1}
    \controlled{\dotgate}{7.50,2}{1}
\end{qcircuit}
}
\\\\[-1ex]
}
\def\commXC{
\m{\begin{qcircuit}[scale=0.5]
    \grid{4.00}{0}
    \gate{$X$}{1.25,0}
    \cgate{1}{2.75}{0}
\end{qcircuit}
}
&=&
\m{\begin{qcircuit}[scale=0.5]
    \grid{2.50}{0}
    \cgate{2}{1.25}{0}
\end{qcircuit}
}
\\\\[-1ex]
\m{\begin{qcircuit}[scale=0.5]
    \grid{4.00}{0}
    \gate{$X$}{1.25,0}
    \cgate{2}{2.75}{0}
\end{qcircuit}
}
&=&
\m{\begin{qcircuit}[scale=0.5]
    \grid{2.50}{0}
    \cgate{1}{1.25}{0}
\end{qcircuit}
}
\\\\[-1ex]
}
\def\commSC{
\m{\begin{qcircuit}[scale=0.5]
    \grid{4.00}{0}
    \gate{$S$}{1.25,0}
    \cgate{1}{2.75}{0}
\end{qcircuit}
}
&=&
\m{\begin{qcircuit}[scale=0.5]
    \grid{4.00}{0}
    \cgate{1}{1.25}{0}
    \gate{$S$}{2.75,0}
\end{qcircuit}
}
\\\\[-1ex]
\m{\begin{qcircuit}[scale=0.5]
    \grid{4.00}{0}
    \gate{$S$}{1.25,0}
    \cgate{2}{2.75}{0}
\end{qcircuit}
}
&=&
\m{\begin{qcircuit}[scale=0.5]
    \grid{7.00}{0}
    \cgate{2}{1.25}{0}
    \gate{$S$}{2.75,0}
    \gate{$S$}{4.25,0}
    \gate{$S$}{5.75,0}
\end{qcircuit}
}
\cdot \omega^{2}
\\\\[-1ex]
}
\def\commZZCI{
\m{\begin{qcircuit}[scale=0.5]
    \grid{3.50}{0,1}
    \controlled{\dotgate}{1.00,1}{0}
    \cgate{1}{2.25}{1}
\end{qcircuit}
}
&=&
\m{\begin{qcircuit}[scale=0.5]
    \grid{3.50}{0,1}
    \cgate{1}{1.25}{1}
    \controlled{\dotgate}{2.50,1}{0}
\end{qcircuit}
}
\\\\[-1ex]
\m{\begin{qcircuit}[scale=0.5]
    \grid{3.50}{0,1}
    \controlled{\dotgate}{1.00,1}{0}
    \cgate{2}{2.25}{1}
\end{qcircuit}
}
&=&
\m{\begin{qcircuit}[scale=0.5]
    \grid{6.50}{0,1}
    \cgate{2}{1.25}{1}
    \controlled{\dotgate}{2.50,1}{0}
    \gate{$S$}{3.75,0}
    \gate{$S$}{5.25,0}
\end{qcircuit}
}
\\\\[-1ex]
}
\def\commZZIIBBBBI{
\m{\begin{qcircuit}[scale=0.5]
    \grid{5.00}{0,1,2}
    \controlled{\dotgate}{1.00,2}{1}
    \bgate{1}{2.25}{0}{1}
    \bgate{1}{3.75}{1}{2}
\end{qcircuit}
}
&=&
\m{\begin{qcircuit}[scale=0.5]
    \grid{8.00}{0,1,2}
    \bgate{1}{1.25}{0}{1}
    \bgate{1}{2.75}{1}{2}
    \gate{$H$}{4.25,1}
    \gate{$H$}{4.25,0}
    \controlled{\dotgate}{5.50,1}{0}
    \gate{$H$}{6.75,1}
    \gate{$H$}{6.75,0}
\end{qcircuit}
}
\\\\[-1ex]
\m{\begin{qcircuit}[scale=0.5]
    \grid{5.00}{0,1,2}
    \controlled{\dotgate}{1.00,2}{1}
    \bgate{1}{2.25}{0}{1}
    \bgate{2}{3.75}{1}{2}
\end{qcircuit}
}
&=&
\m{\begin{qcircuit}[scale=0.5]
    \grid{8.00}{0,1,2}
    \bgate{4}{1.25}{0}{1}
    \bgate{2}{2.75}{1}{2}
    \gate{$H$}{4.25,0}
    \controlled{\dotgate}{5.50,1}{0}
    \gate{$H$}{6.75,0}
\end{qcircuit}
}
\\\\[-1ex]
\m{\begin{qcircuit}[scale=0.5]
    \grid{5.00}{0,1,2}
    \controlled{\dotgate}{1.00,2}{1}
    \bgate{1}{2.25}{0}{1}
    \bgate{3}{3.75}{1}{2}
\end{qcircuit}
}
&=&
\m{\begin{qcircuit}[scale=0.5]
    \grid{14.00}{0,1,2}
    \bgate{4}{1.25}{0}{1}
    \bgate{3}{2.75}{1}{2}
    \gate{$H$}{4.25,0}
    \gate{$H$}{4.25,1}
    \gate{$S$}{5.75,1}
    \gate{$H$}{7.25,1}
    \controlled{\dotgate}{8.50,0}{1}
    \gate{$S$}{9.75,1}
    \gate{$H$}{11.25,1}
    \gate{$S$}{12.75,1}
    \gate{$H$}{9.75,0}
\end{qcircuit}
}
\cdot \omega^{7}
\\\\[-1ex]
\m{\begin{qcircuit}[scale=0.5]
    \grid{5.00}{0,1,2}
    \controlled{\dotgate}{1.00,2}{1}
    \bgate{1}{2.25}{0}{1}
    \bgate{4}{3.75}{1}{2}
\end{qcircuit}
}
&=&
\m{\begin{qcircuit}[scale=0.5]
    \grid{8.00}{0,1,2}
    \bgate{1}{1.25}{0}{1}
    \bgate{4}{2.75}{1}{2}
    \gate{$H$}{4.25,1}
    \gate{$H$}{4.25,0}
    \controlled{\dotgate}{5.50,1}{0}
    \gate{$H$}{6.75,1}
    \gate{$H$}{6.75,0}
\end{qcircuit}
}
\\\\[-1ex]
\m{\begin{qcircuit}[scale=0.5]
    \grid{5.00}{0,1,2}
    \controlled{\dotgate}{1.00,2}{1}
    \bgate{2}{2.25}{0}{1}
    \bgate{1}{3.75}{1}{2}
\end{qcircuit}
}
&=&
\m{\begin{qcircuit}[scale=0.5]
    \grid{8.00}{0,1,2}
    \bgate{2}{1.25}{0}{1}
    \bgate{4}{2.75}{1}{2}
    \gate{$H$}{4.25,1}
    \controlled{\dotgate}{5.50,1}{0}
    \gate{$H$}{6.75,1}
\end{qcircuit}
}
\\\\[-1ex]
\m{\begin{qcircuit}[scale=0.5]
    \grid{5.00}{0,1,2}
    \controlled{\dotgate}{1.00,2}{1}
    \bgate{2}{2.25}{0}{1}
    \bgate{2}{3.75}{1}{2}
\end{qcircuit}
}
&=&
\m{\begin{qcircuit}[scale=0.5]
    \grid{11.00}{0,1,2}
    \bgate{3}{1.25}{0}{1}
    \bgate{3}{2.75}{1}{2}
    \gate{$H$}{4.25,1}
    \gate{$S$}{5.75,1}
    \gate{$H$}{7.25,1}
    \gate{$H$}{4.25,0}
    \gate{$S$}{5.75,0}
    \gate{$H$}{7.25,0}
    \controlled{\dotgate}{8.50,1}{0}
    \gate{$S$}{9.75,1}
    \gate{$S$}{9.75,0}
\end{qcircuit}
}
\cdot \omega^{6}
\\\\[-1ex]
\m{\begin{qcircuit}[scale=0.5]
    \grid{5.00}{0,1,2}
    \controlled{\dotgate}{1.00,2}{1}
    \bgate{2}{2.25}{0}{1}
    \bgate{3}{3.75}{1}{2}
\end{qcircuit}
}
&=&
\m{\begin{qcircuit}[scale=0.5]
    \grid{12.50}{0,1,2}
    \bgate{3}{1.25}{0}{1}
    \bgate{2}{2.75}{1}{2}
    \gate{$H$}{4.25,0}
    \gate{$S$}{5.75,0}
    \gate{$H$}{7.25,0}
    \controlled{\dotgate}{8.50,1}{0}
    \gate{$X$}{4.25,2}
    \gate{$H$}{9.75,1}
    \gate{$S$}{11.25,1}
    \gate{$S$}{9.75,0}
\end{qcircuit}
}
\cdot \omega^{7}
\\\\[-1ex]
\m{\begin{qcircuit}[scale=0.5]
    \grid{5.00}{0,1,2}
    \controlled{\dotgate}{1.00,2}{1}
    \bgate{2}{2.25}{0}{1}
    \bgate{4}{3.75}{1}{2}
\end{qcircuit}
}
&=&
\m{\begin{qcircuit}[scale=0.5]
    \grid{8.00}{0,1,2}
    \bgate{2}{1.25}{0}{1}
    \bgate{1}{2.75}{1}{2}
    \gate{$H$}{4.25,1}
    \controlled{\dotgate}{5.50,1}{0}
    \gate{$H$}{6.75,1}
\end{qcircuit}
}
\\\\[-1ex]
\m{\begin{qcircuit}[scale=0.5]
    \grid{5.00}{0,1,2}
    \controlled{\dotgate}{1.00,2}{1}
    \bgate{3}{2.25}{0}{1}
    \bgate{1}{3.75}{1}{2}
\end{qcircuit}
}
&=&
\m{\begin{qcircuit}[scale=0.5]
    \grid{14.00}{0,1,2}
    \bgate{3}{1.25}{0}{1}
    \bgate{4}{2.75}{1}{2}
    \gate{$H$}{4.25,1}
    \gate{$H$}{4.25,0}
    \gate{$S$}{5.75,0}
    \gate{$H$}{7.25,0}
    \controlled{\dotgate}{8.50,1}{0}
    \gate{$H$}{9.75,1}
    \gate{$S$}{9.75,0}
    \gate{$H$}{11.25,0}
    \gate{$S$}{12.75,0}
\end{qcircuit}
}
\cdot \omega^{7}
\\\\[-1ex]
\m{\begin{qcircuit}[scale=0.5]
    \grid{5.00}{0,1,2}
    \controlled{\dotgate}{1.00,2}{1}
    \bgate{3}{2.25}{0}{1}
    \bgate{2}{3.75}{1}{2}
\end{qcircuit}
}
&=&
\m{\begin{qcircuit}[scale=0.5]
    \grid{12.50}{0,1,2}
    \bgate{2}{1.25}{0}{1}
    \bgate{3}{2.75}{1}{2}
    \gate{$H$}{4.25,1}
    \gate{$S$}{5.75,1}
    \gate{$H$}{7.25,1}
    \controlled{\dotgate}{8.50,1}{0}
    \gate{$X$}{4.25,2}
    \gate{$S$}{9.75,1}
    \gate{$H$}{9.75,0}
    \gate{$S$}{11.25,0}
\end{qcircuit}
}
\cdot \omega^{7}
\\\\[-1ex]
\m{\begin{qcircuit}[scale=0.5]
    \grid{5.00}{0,1,2}
    \controlled{\dotgate}{1.00,2}{1}
    \bgate{3}{2.25}{0}{1}
    \bgate{3}{3.75}{1}{2}
\end{qcircuit}
}
&=&
\m{\begin{qcircuit}[scale=0.5]
    \grid{8.00}{0,1,2}
    \bgate{2}{1.25}{0}{1}
    \bgate{2}{2.75}{1}{2}
    \controlled{\dotgate}{4.00,1}{0}
    \gate{$H$}{5.25,1}
    \gate{$S$}{6.75,1}
    \gate{$H$}{5.25,0}
    \gate{$S$}{6.75,0}
\end{qcircuit}
}
\\\\[-1ex]
\m{\begin{qcircuit}[scale=0.5]
    \grid{5.00}{0,1,2}
    \controlled{\dotgate}{1.00,2}{1}
    \bgate{3}{2.25}{0}{1}
    \bgate{4}{3.75}{1}{2}
\end{qcircuit}
}
&=&
\m{\begin{qcircuit}[scale=0.5]
    \grid{14.00}{0,1,2}
    \bgate{3}{1.25}{0}{1}
    \bgate{1}{2.75}{1}{2}
    \gate{$H$}{4.25,1}
    \gate{$H$}{4.25,0}
    \gate{$S$}{5.75,0}
    \gate{$H$}{7.25,0}
    \controlled{\dotgate}{8.50,1}{0}
    \gate{$H$}{9.75,1}
    \gate{$S$}{9.75,0}
    \gate{$H$}{11.25,0}
    \gate{$S$}{12.75,0}
\end{qcircuit}
}
\cdot \omega^{7}
\\\\[-1ex]
\m{\begin{qcircuit}[scale=0.5]
    \grid{5.00}{0,1,2}
    \controlled{\dotgate}{1.00,2}{1}
    \bgate{4}{2.25}{0}{1}
    \bgate{1}{3.75}{1}{2}
\end{qcircuit}
}
&=&
\m{\begin{qcircuit}[scale=0.5]
    \grid{8.00}{0,1,2}
    \bgate{4}{1.25}{0}{1}
    \bgate{1}{2.75}{1}{2}
    \gate{$H$}{4.25,1}
    \gate{$H$}{4.25,0}
    \controlled{\dotgate}{5.50,1}{0}
    \gate{$H$}{6.75,1}
    \gate{$H$}{6.75,0}
\end{qcircuit}
}
\\\\[-1ex]
\m{\begin{qcircuit}[scale=0.5]
    \grid{5.00}{0,1,2}
    \controlled{\dotgate}{1.00,2}{1}
    \bgate{4}{2.25}{0}{1}
    \bgate{2}{3.75}{1}{2}
\end{qcircuit}
}
&=&
\m{\begin{qcircuit}[scale=0.5]
    \grid{8.00}{0,1,2}
    \bgate{1}{1.25}{0}{1}
    \bgate{2}{2.75}{1}{2}
    \gate{$H$}{4.25,0}
    \controlled{\dotgate}{5.50,1}{0}
    \gate{$H$}{6.75,0}
\end{qcircuit}
}
\\\\[-1ex]
\m{\begin{qcircuit}[scale=0.5]
    \grid{5.00}{0,1,2}
    \controlled{\dotgate}{1.00,2}{1}
    \bgate{4}{2.25}{0}{1}
    \bgate{3}{3.75}{1}{2}
\end{qcircuit}
}
&=&
\m{\begin{qcircuit}[scale=0.5]
    \grid{14.00}{0,1,2}
    \bgate{1}{1.25}{0}{1}
    \bgate{3}{2.75}{1}{2}
    \gate{$H$}{4.25,0}
    \gate{$H$}{4.25,1}
    \gate{$S$}{5.75,1}
    \gate{$H$}{7.25,1}
    \controlled{\dotgate}{8.50,0}{1}
    \gate{$S$}{9.75,1}
    \gate{$H$}{11.25,1}
    \gate{$S$}{12.75,1}
    \gate{$H$}{9.75,0}
\end{qcircuit}
}
\cdot \omega^{7}
\\\\[-1ex]
\m{\begin{qcircuit}[scale=0.5]
    \grid{5.00}{0,1,2}
    \controlled{\dotgate}{1.00,2}{1}
    \bgate{4}{2.25}{0}{1}
    \bgate{4}{3.75}{1}{2}
\end{qcircuit}
}
&=&
\m{\begin{qcircuit}[scale=0.5]
    \grid{8.00}{0,1,2}
    \bgate{4}{1.25}{0}{1}
    \bgate{4}{2.75}{1}{2}
    \gate{$H$}{4.25,1}
    \gate{$H$}{4.25,0}
    \controlled{\dotgate}{5.50,1}{0}
    \gate{$H$}{6.75,1}
    \gate{$H$}{6.75,0}
\end{qcircuit}
}
\\\\[-1ex]
}
\def\altIHDD{
\m{\begin{qcircuit}[scale=0.5]
    \grid{4.00}{0,1}
    \gate{$H$}{1.25,0}
    \dgate{1}{2.75}{0}{1}
\end{qcircuit}
}
&=&
\m{\begin{qcircuit}[scale=0.5]
    \grid{4.00}{0,1}
    \dgate{1}{1.25}{0}{1}
    \gate{$H$}{2.75,1}
\end{qcircuit}
}
\\\\[-1ex]
\m{\begin{qcircuit}[scale=0.5]
    \grid{4.00}{0,1}
    \gate{$H$}{1.25,0}
    \dgate{2}{2.75}{0}{1}
\end{qcircuit}
}
&=&
\m{\begin{qcircuit}[scale=0.5]
    \grid{2.50}{0,1}
    \dgate{4}{1.25}{0}{1}
\end{qcircuit}
}
\\\\[-1ex]
\m{\begin{qcircuit}[scale=0.5]
    \grid{4.00}{0,1}
    \gate{$H$}{1.25,0}
    \dgate{3}{2.75}{0}{1}
\end{qcircuit}
}
&=&
\m{\begin{qcircuit}[scale=0.5]
    \grid{10.00}{0,1}
    \dgate{3}{1.25}{0}{1}
    \gate{$S$}{2.75,1}
    \gate{$S$}{4.25,1}
    \gate{$S$}{5.75,1}
    \gate{$H$}{7.25,1}
    \gate{$S$}{8.75,1}
    \gate{$S$}{2.75,0}
    \gate{$S$}{4.25,0}
\end{qcircuit}
}
\\\\[-1ex]
\m{\begin{qcircuit}[scale=0.5]
    \grid{4.00}{0,1}
    \gate{$H$}{1.25,0}
    \dgate{4}{2.75}{0}{1}
\end{qcircuit}
}
&=&
\m{\begin{qcircuit}[scale=0.5]
    \grid{2.50}{0,1}
    \dgate{2}{1.25}{0}{1}
\end{qcircuit}
}
\\\\[-1ex]
}
\def\altISDD{
\m{\begin{qcircuit}[scale=0.5]
    \grid{4.00}{0,1}
    \gate{$S$}{1.25,0}
    \dgate{1}{2.75}{0}{1}
\end{qcircuit}
}
&=&
\m{\begin{qcircuit}[scale=0.5]
    \grid{7.00}{0,1}
    \dgate{1}{1.25}{0}{1}
    \gate{$H$}{2.75,1}
    \gate{$S$}{4.25,1}
    \gate{$H$}{5.75,1}
\end{qcircuit}
}
\\\\[-1ex]
\m{\begin{qcircuit}[scale=0.5]
    \grid{4.00}{0,1}
    \gate{$S$}{1.25,0}
    \dgate{2}{2.75}{0}{1}
\end{qcircuit}
}
&=&
\m{\begin{qcircuit}[scale=0.5]
    \grid{10.00}{0,1}
    \dgate{3}{1.25}{0}{1}
    \gate{$S$}{2.75,1}
    \gate{$S$}{4.25,1}
    \gate{$S$}{5.75,1}
    \gate{$H$}{7.25,1}
    \gate{$S$}{8.75,1}
    \gate{$S$}{2.75,0}
    \gate{$S$}{4.25,0}
\end{qcircuit}
}
\\\\[-1ex]
\m{\begin{qcircuit}[scale=0.5]
    \grid{4.00}{0,1}
    \gate{$S$}{1.25,0}
    \dgate{3}{2.75}{0}{1}
\end{qcircuit}
}
&=&
\m{\begin{qcircuit}[scale=0.5]
    \grid{7.00}{0,1}
    \dgate{2}{1.25}{0}{1}
    \gate{$S$}{2.75,1}
    \gate{$H$}{4.25,1}
    \gate{$S$}{5.75,1}
\end{qcircuit}
}
\\\\[-1ex]
\m{\begin{qcircuit}[scale=0.5]
    \grid{4.00}{0,1}
    \gate{$S$}{1.25,0}
    \dgate{4}{2.75}{0}{1}
\end{qcircuit}
}
&=&
\m{\begin{qcircuit}[scale=0.5]
    \grid{7.00}{0,1}
    \dgate{4}{1.25}{0}{1}
    \gate{$H$}{2.75,1}
    \gate{$S$}{4.25,1}
    \gate{$H$}{5.75,1}
\end{qcircuit}
}
\\\\[-1ex]
}
\def\altSIDD{
\m{\begin{qcircuit}[scale=0.5]
    \grid{4.00}{0,1}
    \gate{$S$}{1.25,1}
    \dgate{1}{2.75}{0}{1}
\end{qcircuit}
}
&=&
\m{\begin{qcircuit}[scale=0.5]
    \grid{4.00}{0,1}
    \dgate{1}{1.25}{0}{1}
    \gate{$S$}{2.75,0}
\end{qcircuit}
}
\\\\[-1ex]
\m{\begin{qcircuit}[scale=0.5]
    \grid{4.00}{0,1}
    \gate{$S$}{1.25,1}
    \dgate{2}{2.75}{0}{1}
\end{qcircuit}
}
&=&
\m{\begin{qcircuit}[scale=0.5]
    \grid{4.00}{0,1}
    \dgate{2}{1.25}{0}{1}
    \gate{$S$}{2.75,0}
\end{qcircuit}
}
\\\\[-1ex]
\m{\begin{qcircuit}[scale=0.5]
    \grid{4.00}{0,1}
    \gate{$S$}{1.25,1}
    \dgate{3}{2.75}{0}{1}
\end{qcircuit}
}
&=&
\m{\begin{qcircuit}[scale=0.5]
    \grid{4.00}{0,1}
    \dgate{3}{1.25}{0}{1}
    \gate{$S$}{2.75,0}
\end{qcircuit}
}
\\\\[-1ex]
\m{\begin{qcircuit}[scale=0.5]
    \grid{4.00}{0,1}
    \gate{$S$}{1.25,1}
    \dgate{4}{2.75}{0}{1}
\end{qcircuit}
}
&=&
\m{\begin{qcircuit}[scale=0.5]
    \grid{4.00}{0,1}
    \dgate{4}{1.25}{0}{1}
    \gate{$S$}{2.75,0}
\end{qcircuit}
}
\\\\[-1ex]
}
\def\altZZDD{
\m{\begin{qcircuit}[scale=0.5]
    \grid{3.50}{0,1}
    \controlled{\dotgate}{1.00,1}{0}
    \dgate{1}{2.25}{0}{1}
\end{qcircuit}
}
&=&
\m{\begin{qcircuit}[scale=0.5]
    \grid{2.50}{0,1}
    \dgate{4}{1.25}{0}{1}
\end{qcircuit}
}
\\\\[-1ex]
\m{\begin{qcircuit}[scale=0.5]
    \grid{3.50}{0,1}
    \controlled{\dotgate}{1.00,1}{0}
    \dgate{2}{2.25}{0}{1}
\end{qcircuit}
}
&=&
\m{\begin{qcircuit}[scale=0.5]
    \grid{8.50}{0,1}
    \dgate{3}{1.25}{0}{1}
    \gate{$S$}{2.75,1}
    \gate{$S$}{4.25,1}
    \gate{$S$}{5.75,1}
    \gate{$H$}{7.25,1}
    \gate{$S$}{2.75,0}
    \gate{$S$}{4.25,0}
    \gate{$S$}{5.75,0}
\end{qcircuit}
}
\\\\[-1ex]
\m{\begin{qcircuit}[scale=0.5]
    \grid{3.50}{0,1}
    \controlled{\dotgate}{1.00,1}{0}
    \dgate{3}{2.25}{0}{1}
\end{qcircuit}
}
&=&
\m{\begin{qcircuit}[scale=0.5]
    \grid{5.50}{0,1}
    \dgate{2}{1.25}{0}{1}
    \gate{$H$}{2.75,1}
    \gate{$S$}{4.25,1}
    \gate{$S$}{2.75,0}
\end{qcircuit}
}
\\\\[-1ex]
\m{\begin{qcircuit}[scale=0.5]
    \grid{3.50}{0,1}
    \controlled{\dotgate}{1.00,1}{0}
    \dgate{4}{2.25}{0}{1}
\end{qcircuit}
}
&=&
\m{\begin{qcircuit}[scale=0.5]
    \grid{2.50}{0,1}
    \dgate{1}{1.25}{0}{1}
\end{qcircuit}
}
\\\\[-1ex]
}
\def\altSE{
\m{\begin{qcircuit}[scale=0.5]
    \grid{4.00}{0}
    \gate{$S$}{1.25,0}
    \egate{1}{2.75}{0}
\end{qcircuit}
}
&=&
\m{\begin{qcircuit}[scale=0.5]
    \grid{2.50}{0}
    \egate{2}{1.25}{0}
\end{qcircuit}
}
\\\\[-1ex]
\m{\begin{qcircuit}[scale=0.5]
    \grid{4.00}{0}
    \gate{$S$}{1.25,0}
    \egate{2}{2.75}{0}
\end{qcircuit}
}
&=&
\m{\begin{qcircuit}[scale=0.5]
    \grid{2.50}{0}
    \egate{3}{1.25}{0}
\end{qcircuit}
}
\\\\[-1ex]
\m{\begin{qcircuit}[scale=0.5]
    \grid{4.00}{0}
    \gate{$S$}{1.25,0}
    \egate{3}{2.75}{0}
\end{qcircuit}
}
&=&
\m{\begin{qcircuit}[scale=0.5]
    \grid{2.50}{0}
    \egate{4}{1.25}{0}
\end{qcircuit}
}
\\\\[-1ex]
\m{\begin{qcircuit}[scale=0.5]
    \grid{4.00}{0}
    \gate{$S$}{1.25,0}
    \egate{4}{2.75}{0}
\end{qcircuit}
}
&=&
\m{\begin{qcircuit}[scale=0.5]
    \grid{2.50}{0}
    \egate{1}{1.25}{0}
\end{qcircuit}
}
\\\\[-1ex]
}
\def\altIZZDDIIDD{
\m{\begin{qcircuit}[scale=0.5]
    \grid{5.00}{0,1,2}
    \controlled{\dotgate}{1.00,1}{0}
    \dgate{1}{2.25}{1}{2}
    \dgate{1}{3.75}{0}{1}
\end{qcircuit}
}
&=&
\m{\begin{qcircuit}[scale=0.5]
    \grid{8.00}{0,1,2}
    \dgate{1}{1.25}{1}{2}
    \dgate{1}{2.75}{0}{1}
    \gate{$H$}{4.25,2}
    \gate{$H$}{4.25,1}
    \controlled{\dotgate}{5.50,2}{1}
    \gate{$H$}{6.75,2}
    \gate{$H$}{6.75,1}
\end{qcircuit}
}
\\\\[-1ex]
\m{\begin{qcircuit}[scale=0.5]
    \grid{5.00}{0,1,2}
    \controlled{\dotgate}{1.00,1}{0}
    \dgate{1}{2.25}{1}{2}
    \dgate{2}{3.75}{0}{1}
\end{qcircuit}
}
&=&
\m{\begin{qcircuit}[scale=0.5]
    \grid{8.00}{0,1,2}
    \dgate{4}{1.25}{1}{2}
    \dgate{2}{2.75}{0}{1}
    \gate{$H$}{4.25,2}
    \controlled{\dotgate}{5.50,2}{1}
    \gate{$H$}{6.75,2}
\end{qcircuit}
}
\\\\[-1ex]
\m{\begin{qcircuit}[scale=0.5]
    \grid{5.00}{0,1,2}
    \controlled{\dotgate}{1.00,1}{0}
    \dgate{1}{2.25}{1}{2}
    \dgate{3}{3.75}{0}{1}
\end{qcircuit}
}
&=&
\m{\begin{qcircuit}[scale=0.5]
    \grid{14.00}{0,1,2}
    \dgate{4}{1.25}{1}{2}
    \dgate{3}{2.75}{0}{1}
    \gate{$H$}{4.25,2}
    \gate{$H$}{4.25,1}
    \gate{$S$}{5.75,1}
    \gate{$H$}{7.25,1}
    \controlled{\dotgate}{8.50,2}{1}
    \gate{$H$}{9.75,2}
    \gate{$S$}{9.75,1}
    \gate{$H$}{11.25,1}
    \gate{$S$}{12.75,1}
\end{qcircuit}
}
\cdot \omega^{7}
\\\\[-1ex]
\m{\begin{qcircuit}[scale=0.5]
    \grid{5.00}{0,1,2}
    \controlled{\dotgate}{1.00,1}{0}
    \dgate{1}{2.25}{1}{2}
    \dgate{4}{3.75}{0}{1}
\end{qcircuit}
}
&=&
\m{\begin{qcircuit}[scale=0.5]
    \grid{8.00}{0,1,2}
    \dgate{1}{1.25}{1}{2}
    \dgate{4}{2.75}{0}{1}
    \gate{$H$}{4.25,2}
    \gate{$H$}{4.25,1}
    \controlled{\dotgate}{5.50,2}{1}
    \gate{$H$}{6.75,2}
    \gate{$H$}{6.75,1}
\end{qcircuit}
}
\\\\[-1ex]
\m{\begin{qcircuit}[scale=0.5]
    \grid{5.00}{0,1,2}
    \controlled{\dotgate}{1.00,1}{0}
    \dgate{2}{2.25}{1}{2}
    \dgate{1}{3.75}{0}{1}
\end{qcircuit}
}
&=&
\m{\begin{qcircuit}[scale=0.5]
    \grid{8.00}{0,1,2}
    \dgate{2}{1.25}{1}{2}
    \dgate{4}{2.75}{0}{1}
    \gate{$H$}{4.25,1}
    \controlled{\dotgate}{5.50,2}{1}
    \gate{$H$}{6.75,1}
\end{qcircuit}
}
\\\\[-1ex]
\m{\begin{qcircuit}[scale=0.5]
    \grid{5.00}{0,1,2}
    \controlled{\dotgate}{1.00,1}{0}
    \dgate{2}{2.25}{1}{2}
    \dgate{2}{3.75}{0}{1}
\end{qcircuit}
}
&=&
\m{\begin{qcircuit}[scale=0.5]
    \grid{11.00}{0,1,2}
    \dgate{3}{1.25}{1}{2}
    \dgate{3}{2.75}{0}{1}
    \gate{$H$}{4.25,2}
    \gate{$S$}{5.75,2}
    \gate{$H$}{7.25,2}
    \gate{$H$}{4.25,1}
    \gate{$S$}{5.75,1}
    \gate{$H$}{7.25,1}
    \controlled{\dotgate}{8.50,2}{1}
    \gate{$S$}{9.75,2}
    \gate{$S$}{9.75,1}
\end{qcircuit}
}
\cdot \omega^{6}
\\\\[-1ex]
\m{\begin{qcircuit}[scale=0.5]
    \grid{5.00}{0,1,2}
    \controlled{\dotgate}{1.00,1}{0}
    \dgate{2}{2.25}{1}{2}
    \dgate{3}{3.75}{0}{1}
\end{qcircuit}
}
&=&
\m{\begin{qcircuit}[scale=0.5]
    \grid{12.50}{0,1,2}
    \dgate{3}{1.25}{1}{2}
    \dgate{2}{2.75}{0}{1}
    \gate{$H$}{4.25,2}
    \gate{$S$}{5.75,2}
    \gate{$H$}{7.25,2}
    \controlled{\dotgate}{8.50,2}{1}
    \gate{$S$}{9.75,2}
    \gate{$H$}{9.75,1}
    \gate{$S$}{11.25,1}
    \gate{$S$}{4.25,0}
    \gate{$S$}{5.75,0}
\end{qcircuit}
}
\cdot \omega^{7}
\\\\[-1ex]
\m{\begin{qcircuit}[scale=0.5]
    \grid{5.00}{0,1,2}
    \controlled{\dotgate}{1.00,1}{0}
    \dgate{2}{2.25}{1}{2}
    \dgate{4}{3.75}{0}{1}
\end{qcircuit}
}
&=&
\m{\begin{qcircuit}[scale=0.5]
    \grid{8.00}{0,1,2}
    \dgate{2}{1.25}{1}{2}
    \dgate{1}{2.75}{0}{1}
    \gate{$H$}{4.25,1}
    \controlled{\dotgate}{5.50,2}{1}
    \gate{$H$}{6.75,1}
\end{qcircuit}
}
\\\\[-1ex]
\m{\begin{qcircuit}[scale=0.5]
    \grid{5.00}{0,1,2}
    \controlled{\dotgate}{1.00,1}{0}
    \dgate{3}{2.25}{1}{2}
    \dgate{1}{3.75}{0}{1}
\end{qcircuit}
}
&=&
\m{\begin{qcircuit}[scale=0.5]
    \grid{14.00}{0,1,2}
    \dgate{3}{1.25}{1}{2}
    \dgate{4}{2.75}{0}{1}
    \gate{$H$}{4.25,1}
    \gate{$H$}{4.25,2}
    \gate{$S$}{5.75,2}
    \gate{$H$}{7.25,2}
    \controlled{\dotgate}{8.50,1}{2}
    \gate{$S$}{9.75,2}
    \gate{$H$}{11.25,2}
    \gate{$S$}{12.75,2}
    \gate{$H$}{9.75,1}
\end{qcircuit}
}
\cdot \omega^{7}
\\\\[-1ex]
\m{\begin{qcircuit}[scale=0.5]
    \grid{5.00}{0,1,2}
    \controlled{\dotgate}{1.00,1}{0}
    \dgate{3}{2.25}{1}{2}
    \dgate{2}{3.75}{0}{1}
\end{qcircuit}
}
&=&
\m{\begin{qcircuit}[scale=0.5]
    \grid{12.50}{0,1,2}
    \dgate{2}{1.25}{1}{2}
    \dgate{3}{2.75}{0}{1}
    \gate{$H$}{4.25,1}
    \gate{$S$}{5.75,1}
    \gate{$H$}{7.25,1}
    \controlled{\dotgate}{8.50,2}{1}
    \gate{$H$}{9.75,2}
    \gate{$S$}{11.25,2}
    \gate{$S$}{9.75,1}
    \gate{$S$}{4.25,0}
    \gate{$S$}{5.75,0}
\end{qcircuit}
}
\cdot \omega^{7}
\\\\[-1ex]
\m{\begin{qcircuit}[scale=0.5]
    \grid{5.00}{0,1,2}
    \controlled{\dotgate}{1.00,1}{0}
    \dgate{3}{2.25}{1}{2}
    \dgate{3}{3.75}{0}{1}
\end{qcircuit}
}
&=&
\m{\begin{qcircuit}[scale=0.5]
    \grid{8.00}{0,1,2}
    \dgate{2}{1.25}{1}{2}
    \dgate{2}{2.75}{0}{1}
    \controlled{\dotgate}{4.00,2}{1}
    \gate{$H$}{5.25,2}
    \gate{$S$}{6.75,2}
    \gate{$H$}{5.25,1}
    \gate{$S$}{6.75,1}
\end{qcircuit}
}
\\\\[-1ex]
\m{\begin{qcircuit}[scale=0.5]
    \grid{5.00}{0,1,2}
    \controlled{\dotgate}{1.00,1}{0}
    \dgate{3}{2.25}{1}{2}
    \dgate{4}{3.75}{0}{1}
\end{qcircuit}
}
&=&
\m{\begin{qcircuit}[scale=0.5]
    \grid{14.00}{0,1,2}
    \dgate{3}{1.25}{1}{2}
    \dgate{1}{2.75}{0}{1}
    \gate{$H$}{4.25,1}
    \gate{$H$}{4.25,2}
    \gate{$S$}{5.75,2}
    \gate{$H$}{7.25,2}
    \controlled{\dotgate}{8.50,1}{2}
    \gate{$S$}{9.75,2}
    \gate{$H$}{11.25,2}
    \gate{$S$}{12.75,2}
    \gate{$H$}{9.75,1}
\end{qcircuit}
}
\cdot \omega^{7}
\\\\[-1ex]
\m{\begin{qcircuit}[scale=0.5]
    \grid{5.00}{0,1,2}
    \controlled{\dotgate}{1.00,1}{0}
    \dgate{4}{2.25}{1}{2}
    \dgate{1}{3.75}{0}{1}
\end{qcircuit}
}
&=&
\m{\begin{qcircuit}[scale=0.5]
    \grid{8.00}{0,1,2}
    \dgate{4}{1.25}{1}{2}
    \dgate{1}{2.75}{0}{1}
    \gate{$H$}{4.25,2}
    \gate{$H$}{4.25,1}
    \controlled{\dotgate}{5.50,2}{1}
    \gate{$H$}{6.75,2}
    \gate{$H$}{6.75,1}
\end{qcircuit}
}
\\\\[-1ex]
\m{\begin{qcircuit}[scale=0.5]
    \grid{5.00}{0,1,2}
    \controlled{\dotgate}{1.00,1}{0}
    \dgate{4}{2.25}{1}{2}
    \dgate{2}{3.75}{0}{1}
\end{qcircuit}
}
&=&
\m{\begin{qcircuit}[scale=0.5]
    \grid{8.00}{0,1,2}
    \dgate{1}{1.25}{1}{2}
    \dgate{2}{2.75}{0}{1}
    \gate{$H$}{4.25,2}
    \controlled{\dotgate}{5.50,2}{1}
    \gate{$H$}{6.75,2}
\end{qcircuit}
}
\\\\[-1ex]
\m{\begin{qcircuit}[scale=0.5]
    \grid{5.00}{0,1,2}
    \controlled{\dotgate}{1.00,1}{0}
    \dgate{4}{2.25}{1}{2}
    \dgate{3}{3.75}{0}{1}
\end{qcircuit}
}
&=&
\m{\begin{qcircuit}[scale=0.5]
    \grid{14.00}{0,1,2}
    \dgate{1}{1.25}{1}{2}
    \dgate{3}{2.75}{0}{1}
    \gate{$H$}{4.25,2}
    \gate{$H$}{4.25,1}
    \gate{$S$}{5.75,1}
    \gate{$H$}{7.25,1}
    \controlled{\dotgate}{8.50,2}{1}
    \gate{$H$}{9.75,2}
    \gate{$S$}{9.75,1}
    \gate{$H$}{11.25,1}
    \gate{$S$}{12.75,1}
\end{qcircuit}
}
\cdot \omega^{7}
\\\\[-1ex]
\m{\begin{qcircuit}[scale=0.5]
    \grid{5.00}{0,1,2}
    \controlled{\dotgate}{1.00,1}{0}
    \dgate{4}{2.25}{1}{2}
    \dgate{4}{3.75}{0}{1}
\end{qcircuit}
}
&=&
\m{\begin{qcircuit}[scale=0.5]
    \grid{8.00}{0,1,2}
    \dgate{4}{1.25}{1}{2}
    \dgate{4}{2.75}{0}{1}
    \gate{$H$}{4.25,2}
    \gate{$H$}{4.25,1}
    \controlled{\dotgate}{5.50,2}{1}
    \gate{$H$}{6.75,2}
    \gate{$H$}{6.75,1}
\end{qcircuit}
}
\\\\[-1ex]
}